\documentclass{article}

\usepackage{amsmath,amssymb,amsthm}
\usepackage{lingmacros}
\usepackage{tree-dvips}
\usepackage{graphicx}
\usepackage{xcolor}
\usepackage{cleveref}
\usepackage{subfigure}
\usepackage{ textcomp }

\usepackage{multirow}
\usepackage{booktabs}
\usepackage{wrapfig}

\usepackage{enumitem}

\usepackage{siunitx}
\newtheorem{lemma}{Lemma}

\usepackage{makecell}

\newcommand{\R}{\mathbb{R}}
\newcommand{\Z}{\mathbb{Z}}
\DeclareMathOperator{\Arg}{Arg}
\DeclareMathOperator{\Mem}{Mem}
\DeclareMathOperator{\Act}{A}

\newcommand{\Abs}[1]{\left|#1\right|}

\newcommand{\taufwd}{\tau_{\textnormal{fwd}}}
\newcommand{\taubkwd}{\tau_{\textnormal{bkwd}}}
\newcommand{\taurecomp}{\tau_{\textnormal{recomp}}}

\newif\ifshowall
\showalltrue
\newif\ifarxiv
\arxivtrue

\ifarxiv
\RequirePackage{fullpage}
\usepackage{authblk}
\usepackage[numbers]{natbib}

\title{PipeMare: Asynchronous Pipeline Parallel DNN Training}
\author[1]{Bowen Yang}
\author[1,2]{Jian Zhang}
\author[1]{Jonathon Li}
\author[1,2]{\\Christopher R\'e}
\author[1,2]{Christopher R. Aberger}
\author[1,3]{Christopher De Sa}
\affil[1]{SambaNova Systems}
\affil[2]{Department of Computer Science, Stanford University}
\affil[3]{Department of Computer Science, Cornell University}
\affil[ ]{\texttt{bowen.yang@sambanovasystems.com}, \texttt{\{zjian, jlli\}@stanford.edu},}
\affil[ ]{\texttt{christopher.aberger@sambanovasystems.com}, \texttt{\{chrismre\}@cs.stanford.edu},}
\affil[ ]{\texttt{cdesa@cs.cornell.edu}}

\date{}

\else

\usepackage{icml2020}

\icmltitlerunning{PipeMare: Asynchronous Pipeline Parallel DNN Training}

\fi

\begin{document}

\ifarxiv

\maketitle

\begin{abstract}
Pipeline parallelism (PP) when training neural networks enables larger models to be partitioned spatially, leading to both lower network communication and overall higher hardware utilization.
Unfortunately, to preserve the statistical efficiency of sequential training, existing PP techniques sacrifice hardware efficiency by decreasing pipeline utilization or incurring extra memory costs.
In this paper, we investigate to what extent these sacrifices are necessary.
We devise PipeMare, a simple yet robust training method that tolerates asynchronous updates during PP execution without sacrificing utilization or memory, which allows efficient use of fine-grained pipeline parallelism.
Concretely, when tested on ResNet and Transformer networks, asynchrony enables PipeMare to use up to $2.7\times$ less memory or get $4.3\times$ higher pipeline utilization, with similar model quality, when compared to state-of-the-art synchronous PP training techniques.
 \end{abstract}

\else

\twocolumn[
\icmltitle{PipeMare: Asynchronous Pipeline Parallel DNN Training}
\begin{icmlauthorlist}
\icmlauthor{Bowen Yang}{to}
\icmlauthor{Jian Zhang}{to}
\icmlauthor{Jonathan Li}{stan}
\icmlauthor{Christopher R\'e}{to, stan}
\icmlauthor{Christopher R. Aberger}{to}
\icmlauthor{Christopher De Sa}{to, corn}
\end{icmlauthorlist} 

\icmlaffiliation{to}{SambaNova Systems}
\icmlaffiliation{stan}{Stanford University}
\icmlaffiliation{corn}{Cornell University}

\icmlcorrespondingauthor{Bowen Yang}{bowen.yang@sambanovasystems.com}

\icmlkeywords{Machine Learning, icML}

\vskip 0.3in

\begin{abstract}
Pipeline parallelism (PP) when training neural networks enables larger models to be partitioned spatially, leading to both lower network communication and overall higher hardware utilization.
Unfortunately, to preserve the statistical efficiency of sequential training, existing PP techniques sacrifice hardware efficiency by decreasing pipeline utilization or incurring extra memory costs.
In this paper, we investigate to what extent these sacrifices are necessary.
We devise PipeMare, a simple yet robust training method that tolerates asynchronous updates during PP execution without sacrificing utilization or memory, which allows efficient use of fine-grained pipeline parallelism.
Concretely, when tested on ResNet and Transformer networks, asynchrony enables PipeMare to use up to $2.7\times$ less memory or get $4.3\times$ higher pipeline utilization, with similar model quality, when compared to state-of-the-art synchronous PP training techniques.
 \end{abstract}
]

\printAffiliationsAndNotice{}  %

\fi

\newcommand{\todo}[1]{{\color{red} {#1}}}
\newcommand{\update}[1]{{{#1}}}

\section{Introduction}
Recently there has been a explosion of interest in hardware chips designed 
for training deep neural networks 
\cite{cerebras,graphcore,jouppi2017datacenter,habana}.
These works rethink how computations are mapped to hardware, which can result in huge speedups.
One of the central ideas that has emerged out of this effort
is that model parallelism can be leveraged in place of, or in combination with, data parallelism.
Model parallelism entails partitioning neural network layers spatially across hardware 
resources while pipelining the computation between them. 
Training a neural 
network in this model-parallel fashion is called \emph{pipeline parallelism} (PP).

There are several benefits of PP over traditional data
parallel execution.
First, it \emph{eliminates context switching}.
Without pipeline parallelism, GPUs run neural network training on a kernel-by-kernel basis. Each new low-level operator or kernel results in a context switch: it must be dynamically dispatched from the CPU to the GPU, which can incur time delays. Instead, with pipeline parallelism, context switching is no longer necessary. Operators are spatially fixed across compute resources, and the entire computation graph is able to run in a single context without dynamic dispatching. 
Second, PP \emph{alleviates the accelerator memory bottleneck}.
When training deep neural networks (DNNs) with a kernel-by-kernel accelerator, weights must continually be marshalled back and forth to main memory.
This can be a major bottleneck, especially in distributed data-parallel settings where weights must be replicated across all devices in the system.
This problem is only getting worse as state-of-the-art models are continually growing across domains, requiring them to be distributed \cite{roberta,inception,xlnet}.
PP alleviates this memory bottleneck in a distributed setting
by splitting up, rather than replicating, the weights across accelerators.
Third, prior work has shown that PP can \emph{reduce network bandwidth} by up to 95\%.
Reducing the pressure of bandwidth demands is particularly important for data-parallel distributed systems, since in such systems the communication that
occurs between the devices is proportional to the number of parameters \cite{harlap2018pipedream}. 

Despite these many hardware efficiency benefits of pipeline parallelism, 
existing PP techniques sacrifice hardware efficiency to preserve a property called ``synchronous execution,'' which is believed to be necessary to maintain statistical efficiency (e.g. classification accuracy).
\emph{Synchronous execution} in this context means that the 
weights used for computation during forward propagation are the same as those
used to compute the gradients during backward propagation (as if the gradient were computed in one step).
Existing approaches preserve synchronous execution by trading off pipeline utilization (by adding bubbles into the pipeline, which underutilizes the hardware) and/or memory (by storing additional weight copies for microbatching) \cite{huang2018gpipe,harlap2018pipedream}.
Importantly, these costs increase with the pipeline depth (as illustrated in \Cref{fig:pipeline_extremes}) even though the intention of increasing the pipeline depth is to improve throughput.
This poses a massive challenge for the type of high-depth fine-grained PP that could run very efficiently on new hardware accelerators: as a result, previous PP techniques have focused purely on the lower-depth distributed setting.
Motivated both by enabling PP on new hardware accelerators and improving efficiency in the distributed setting, in this paper we study how to remove hardware overheads during pipeline parallel training by revisiting the fundamental question: \emph{is preserving synchronous execution necessary during neural network training?} 
Our contributions and outline are as follows.

\begin{itemize}[nosep, itemindent=24pt,leftmargin=0pt]
\item In \Cref{sec:prelim}, we introduce a model for \emph{asynchronous} 
pipeline-parallel training, that, by eschewing synchronous execution, maximizes hardware efficiency by avoiding both pipeline bubbles and substantial memory increases.
\item  Using this model, in \Cref{sec:theory} we propose \textbf{PipeMare}, a system of two techniques to improve the statistical efficiency of asynchronous pipeline training. 
\item In \Cref{sec:experiments}, we evaluate PipeMare on ResNet50 and Transformer models. PipeMare can achieve competitive model accuracy with better hardware utilization than previous approaches (GPipe and PipeDream). 
\end{itemize}
We hope our work contributes to the design of future hardware accelerators which are
capable of realizing the hardware benefits of high-depth fine-grained PP.

\begin{figure}
\centering 
\includegraphics[width=0.32\linewidth]{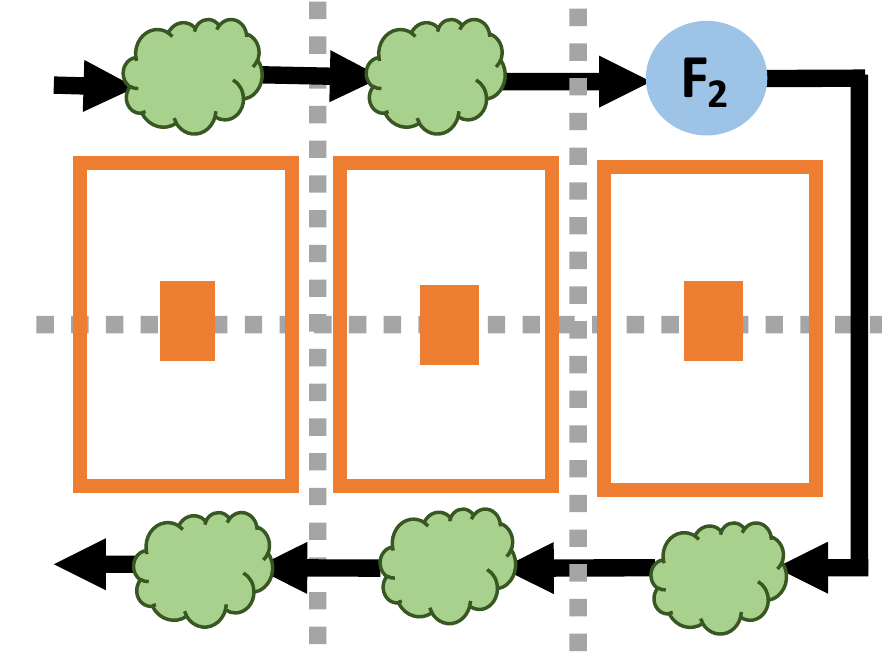}\hfill
\includegraphics[width=0.32\linewidth]{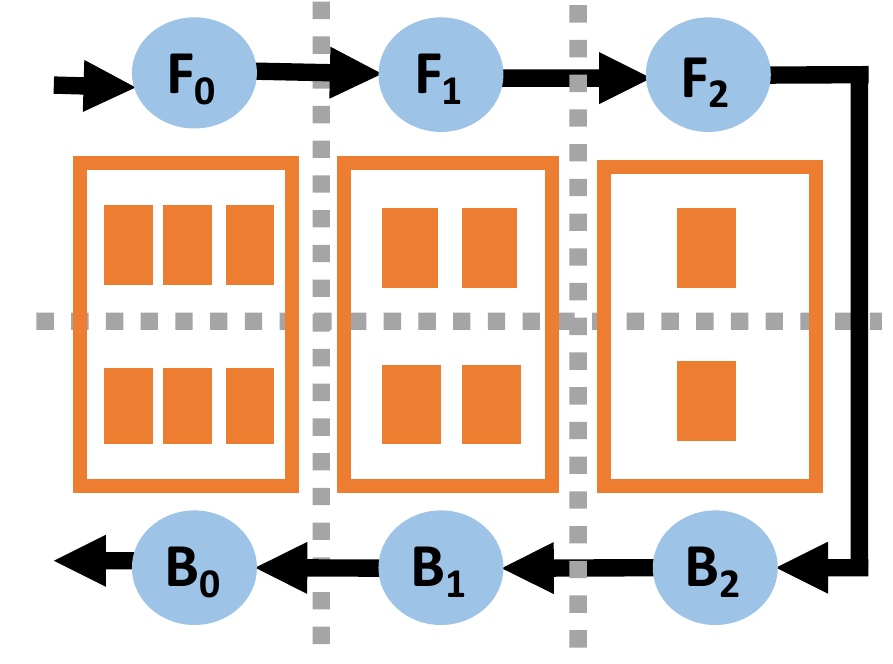}\hfill
\includegraphics[width=0.32\linewidth]{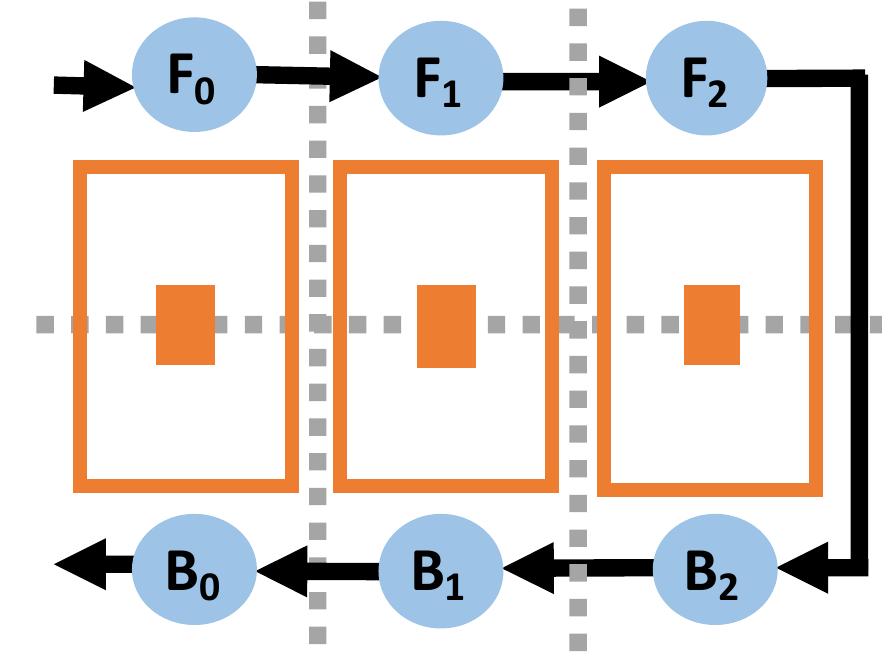} \\
\ifarxiv%
{\makebox[0.32\linewidth][c]{(a) Throughput-Poor}\hfill%
\makebox[0.32\linewidth][c]{(b) Memory-Hungry}\hfill%
\makebox[0.32\linewidth][c]{(c) PipeMare}}
\caption{Different extremes of pipelining modes. Orange squares represent model weight memory, blue circles represent active pipeline compute, green clouds represent pipeline bubbles, and dashed gray lines separate pipeline stages. PipeMare fully utilizes compute while minimizing weight memory footprint.}
\else%
{\small{}(a) Throughput-Poor\hspace{1em}(b) Memory-Hungry\hspace{2em}(c) PipeMare\hfill}
\vspace{-3mm}
\caption{Different extremes of pipelining modes. Orange squares represent model weight memory, blue circles represent active pipeline compute, green clouds represent pipeline bubbles, and dashed gray lines separate pipeline stages. PipeMare fully utilizes compute while minimizing weight memory footprint.}
\vspace{-3em}
\fi%
\label{fig:pipeline_extremes}\
\end{figure}

\subsection{Related Work}

\textbf{PipeDream.}\hspace{1em}PipeDream \cite{harlap2018pipedream} is a PP distributed training
technique used to reduce high computation-to-communication ratios.
PipeDream showed up to 5x speedups in time-to-given-accuracy metrics when 
compared to existing data parallel training techniques.
PipeDream is one type of memory hungry pipelining approach; their core technique
is called weight stashing which maintains an additional copy of the weights for each minibatch 
flowing through the pipeline. This ensures synchronous computation with 
a fixed pipeline delay update. 

\textbf{GPipe.}\hspace{1em}GPipe \cite{huang2018gpipe} is a PP distributed training technique
originally deployed on TPUs \cite{jouppi2017datacenter}. 
GPipe is one type of throughput poor pipelining approach; the core technique used in GPipe is microbatching
to hide the latency from introducing bubbles into the pipeline.
This preserves synchronous execution. 
This approach requires extra activation memory
to preserve synchronous execution across batch boundaries; the authors use gradient 
checkpointing \cite{recompute} to alleviate this memory cost. 
Using these techniques, they show that PP
can enable training larger models than ever on TPUs. In this paper, we focus on leveraging
microbatching to reduce asynchrony, but we also validate that, like GPipe, PipeMare can leverage
gradient checkpointing to reduce activation memory (see \Cref{sec:activation_mem,app:recompute}).

\textbf{Hogwild!}\hspace{1em}Asynchronous training has been studied in several other contexts,
the most well-known of which is Hogwild! \cite{hogwild}. In Hogwild! settings,
as in pipeline-parallel settings, gradients are computed based on delayed versions
of weights.
However, these delays are random and can vary from step to step and weight to weight,
unlike the fixed pipeline delay of the pipeline-parallel setting.
In \Cref{app:hogwild_extension} we extend and apply PipeMare to this setting \cite{recht2011hogwild} showing 
that it can also improve final model accuracies here.

\section{Preliminaries}
\label{sec:prelim}

We formally define a model of pipeline parallelism 
and asynchronous learning that forms the basis for the remainder of this paper.
In \Cref{sec:model} we define the model, and in \Cref{sec:methods} we analyze
the delays, pipeline utilization, and weight memories of GPipe, PipeDream, and PipeMare.

\subsection{Model of Pipeline Parallelism}
\label{sec:model} 

Pipeline-parallel training of a DNN works by decomposing the $L$ layers (or operators) of the neural network
into $P$ pipeline stages, each of which is assigned to a parallel worker (this worker can range from a full distributed machine to a section of silicon on an accelerator).
While processing a minibatch of size $B$, each pipeline stage processes $M$ samples at a time, where $M$ is called the \emph{microbatch size} and $M \le B$.
We use $N$ to represent the number of microbatches in a minibatch (or $N = \lceil\frac{B}{M}\rceil$) and let $i$ represent a pipeline stage.
Layers can be associated with weights: we let $W$ represent the total size of all these weights.
The resulting microbatch gradients are accumulated into gradient buffers, and
weights are updated only at minibatch boundaries.  Previous work 
studied the distributed case where $P \ll L$: we call this \emph{coarse-grained pipeline
parallelism}. Here, we are interested in the case of \emph{fine-grained pipeline parallelism}, where $P \approx L$. 

\textbf{Pipeline Utilization/Throughput.}\hspace{1em} Pipeline utilization ($\operatorname{Util}$) is the 
percentage of pipeline stages that are not idle (stalled) at any given time. In the best case ($P$ active stages), we get 100\% pipeline utilization or $\operatorname{Util}=\lceil\frac{P}{P}=1.0\rceil$. Note that throughput is linearly proportional to $\operatorname{Util}$.

\textbf{Delay.}\hspace{1em}The statistical effect of using pipeline-parallel training is characterized by
the \emph{pipeline delay}: the number of optimizer steps that pass between when the weights are
read to compute a gradient and when that gradient is used to update the weights.
In a standard backpropagation algorithm, each weight is read twice---once in the forward pass,
and again in the backward pass---so there are \emph{two} delay values, $\taufwd$ and $\taubkwd$, which can vary by stage.
Intuitively, $\taufwd$ corresponds to the 
delay between a weight's forward pass and its update. The earlier a pipeline stage, the larger $\taufwd$ value, i.e., $\tau_{\text{fwd},i} \propto (P-i)$ for the $i$th
stage.
Similarly, $\taubkwd$ is the delay between 
a weight's backward pass and its update.
We can write this out formally as
\[
	w_{t+1} = w_t - \alpha \nabla f_t(u_{\text{fwd},t}, u_{\text{bkwd},t})
\]
where $w_t$ are the weight values after $t$ gradient steps, $\nabla f_t$ is the gradient function for the $t$-th minibatch, and $u_{\text{fwd},t}$ and $u_{\text{bkwd},t}$ are the (delayed) values of the weights that are used in the forward and backward passes, respectively, for computing $\nabla f_t$.
The weights $w_t$ can be broken up into weight vectors for each stage: $(w_t)_1$ for stage $1$, $(w_t)_2$ for stage $2$, et cetera, such that $w_t = [(w_t)_1, (w_t)_2, \ldots, (w_t)_P]$.
Concretely, the weight value $(w_t)_i$ for stage $i$ denotes the value of the weights for that stage after $t$ gradient updates have been written to it (this means $w_t$ as a whole is not necessarily the value of the weights in memory at any time $t$),
and the delayed weight values are defined for each pipeline stage $i \in \{1, \ldots, P\}$ as
\[
	\left( u_{\text{fwd},t} \right)_i = \left( w_{t - \tau_{\text{fwd},i}} \right)_i
	\hspace{1em}\text{and}\hspace{1em}
	\left( u_{\text{bkwd},t} \right)_i = \left( w_{t - \tau_{\text{bkwd},i}} \right)_i
\]
where $(\cdot)_i$ denotes selecting the weights for the $i$th stage.
This is a bit of an abuse of notation---here, we are letting $\nabla f_t(u_{\text{fwd},t}, u_{\text{bkwd},t})$ denote \emph{the value of the gradient that would be computed} by the backpropagation algorithm using the weights $u_{\text{fwd},t}$ in the forward pass and weights $u_{\text{bkwd},t}$ in the backward pass.
That is, $\nabla f_t$ is a function of \emph{two} weight vectors, rather than one (as is usual for SGD), because the pipeline-parallel model may use different values for the weights in the forward and backward pass.
Synchronous execution corresponds to the case of $u_{\text{fwd},t} = u_{\text{bkwd},t}$, which requires setting $\taufwd = \taubkwd$.
For the rest of this paper, we will use $\nabla f_t$ with two arguments to denote this backpropagation-with-different-weights gradient, and use $\nabla f_t$ with one argument to denote the ordinary mathematical gradient (under this notation, $\nabla f_t(w,w) = \nabla f_t(w)$ by definition).
Techniques to date have not shown how to train well when $\taufwd - \taubkwd \neq 0$.

\subsection{Pipeline Parallel Training Methods}
\label{sec:methods} 

Using this setup we now analyze the delays, pipeline utilization, and memory usage of 
the two synchronous baseline PP training methods (PipeDream and GPipe), 
and we introduce the setup for our asynchronous method (PipeMare). These results are summarized in
\Cref{tab:taus}.

\textbf{PipeDream.}\hspace{1em}PipeDream has forward delay $\tau_{\text{fwd},i} = \lceil\frac{2(P-i)+1}{N}\rceil$ and uses weight stashing to cache the weights used in the forward pass until they are needed in the backward pass, which allows for
full pipeline efficiency while maintaining synchronous execution $\taufwd = \taubkwd$.
Note that because $\taufwd = \taubkwd$ PipeDream uses same weights for both forward and backward passes, despite having a delayed update.
Unfortunately, this comes at the cost of storing copies of the weights, which uses extra memory of size {\setlength{\abovedisplayskip}{2pt}\setlength{\belowdisplayskip}{2pt}\[ \textstyle \operatorname{Mem} = \sum_{i=0}^{P}|(w)_{i}|{\times}\tau_{\text{fwd},i} = \sum_{i=0}^{P}|(w)_{i}|{\times}\lceil\frac{2(P-i)+1}{N}\rceil. \]}With fine-grained PP $P$ can become large, making the overhead $\operatorname{Mem}$ large, which presents a problem for large models. Because PipeDream's pipeline is fully utilized during training they have a pipeline utilization of $\operatorname{Util} = 1.0$.

\textbf{GPipe.}\hspace{1em}
For GPipe, $\taufwd = \taubkwd = 0$, at the cost of lower pipeline utilization and additional activation memory.
Each pipeline has to be filled and drained at a minibatch boundary to ensure weight synchronization between forward and backward pass, so the average bubble time is $O(\frac{P - 1}{N + P - 1})$ \cite{huang2018gpipe}. Consequently, the pipeline utilization of GPipe is $\frac{N}{N + P - 1}$.
GPipe leverages microbatching (increasing $N > 1$) to reduce the number of bubbles in its pipeline.
GPipe 
does not store any additional weight memory but does store additional memory
for activations. Using the standard technique of gradient checkpointing \cite{recompute}, both PipeMare
and GPipe can reduce their activation memory footprint (see \Cref{app:comp_act_tradeoff,app:recompute}).

\begin{table}[t]
	\ifarxiv%
	\else
	\renewcommand{\arraystretch}{1.25}
	\setlength{\tabcolsep}{2pt}
	\fi
	\centering
	\ifarxiv%
		\begin{tabular}{l c c | c c}
	\else
		\resizebox{\linewidth}{!}{\begin{tabular}{l c c | c c}
	\fi
		\hline
		 & \multicolumn{2}{c|}{Per Stage ($i$)} & \multicolumn{2}{c}{Overall} \\
		 & $\tau_{\text{fwd},i}$ & $\tau_{\text{bkwd},i}$ & $\operatorname{Util}$ & $\Mem$\\ 
		\hline
		PipeDream 		& $\left\lceil \frac{2(P-i)+1}{N} \right\rceil$ & $\tau_{\text{fwd},i}$   &1.0 				   & $\sum_{i=0}^{P}\left|(w)_{i}\right|{\times}\tau_{\text{fwd},i}$\rule{0pt}{2.9ex} \\
		GPipe 			& 0 				   & 0 						&$\frac{N}{N + P - 1}$ &$W=\sum_{i=0}^{P} \left|(w)_{i} \right|$ \\
		PipeMare 		& $\left\lceil \frac{2(P-i)+1}{N} \right\rceil$ & $0$ 	 				&1.0 				   &$W=\sum_{i=0}^{P} \left|(w)_{i} \right|$\rule[-1.8ex]{0pt}{0pt}\\
		\hline
	\ifarxiv
		\end{tabular}
			\caption{Characterization of pipeline parallel training methods. $\taufwd$ and $\taubkwd$ are the pipeline delays for
		model weights in the forwards and backwards pass. $W$ is one copy of the
		weights. $P$ is the number of pipeline stages. $N$ is the number of microbatches in a minibatch. $i$ indexes the pipeline stage. $|(w)_{i}|$ denotes the number of weights in the $i$th layer.}
	\else
		\end{tabular}}
		\caption{Characterization of pipeline parallel training methods. Here, $|(w)_{i}|$ denotes the number of weights in the $i$th layer.}
		\vspace{-5mm}
	\fi
	\label{tab:taus}
\end{table}

\textbf{PipeMare.}\hspace{1em}In PipeMare we let the computation proceed asynchronously: we just compute gradients with whatever weights are in memory at the time we need to use them.
This avoids any need to store extra copies of our model weights ($\operatorname{Mem} = W$) or introduce bubbles into our pipeline ($\operatorname{Util}=1.0$),
because as soon as a pipeline stage has its gradients (accumulated within a full minibatch) the weights are updated.
This means that the forward propagation is done on different weights than those that are used for backpropagation, i.e., $\taufwd \ne \taubkwd$.
Concretely, each layer in our neural network has a fixed forward delay of $\taufwd = \lceil\frac{2(P-i)+1}{N}\rceil$ which is the same $\taufwd$ as PipeDream. On the other hand, since there is no delay between backward pass and weight updates, $\taubkwd = 0$. Similar to GPipe, minimizing the microbatch size $M$ reduces the activation memory usage while also keeping each pipeline stage fully utilized.
Unlike in GPipe, minimizing the microbatch size in PipeMare has the additional benefit of helping to reduce the discrepancy between the forward and backward delays.

\section{PipeMare}
\label{sec:theory}

We design a strategy called \textbf{PipeMare} for asynchronous pipeline-parallel training of deep neural networks.
PipeMare combines two techniques, which we introduce in this section, motivated by theory.
For each technique, we start by modeling a problem we want to address by studying fixed-delay asynchronous gradient descent on a one-dimensional convex quadratic objective.
Even this very simple ``toy'' model has non-trivial behavior, and (as we will see) it exhibits many phenomena of interest seen in more complicated settings, and it motivates techniques to address them that work even for much more complicated objectives (such as for DNNs).
Consider a one-dimensional quadratic objective
$f(w) = \lambda w^2 / 2$
for some fixed curvature $\lambda > 0$.
Suppose that we run fixed-delay asynchronous SGD on this model, using gradient samples of the form
\[
	\nabla f_t(u_{\text{fwd},t}, u_{\text{bkwd},t}) = \lambda u_{\text{fwd},t} - \eta_t = \lambda w_{t-\tau} - \eta_t
\]
where $\eta_t$ is some gradient estimation noise, which we assume is bounded and depends on $t$.
This implicitly assumes that the delays for all the weights are the same and equal to some fixed parameter $\tau = \taufwd$, with no delay discrepancy (we will consider delay discrepancy later in Section~\ref{sec:delay_discrepancy}).
Running SGD in this setting has the update step
\begin{equation}
	w_{t+1} 
	= 
	w_{t} - \alpha \nabla f_t(\cdots)
	=
	w_{t} - \alpha \lambda w_{t-\tau} + \alpha \eta_t. \label{eqnLTI}
\end{equation}

\subsection{Learning rate rescheduling (T1)}
\label{sec:lr_rescheduling}

We theoretically derive our first technique---rescheduling the step size to be inversely proportional to the delay---and evaluate its tradeoffs on some DNN tasks.

\begin{figure}[t]
\centering
	\begin{tabular}{c c}
		\ifarxiv
		\includegraphics[width=0.4\columnwidth, trim=0 0.4cm 0 0cm]{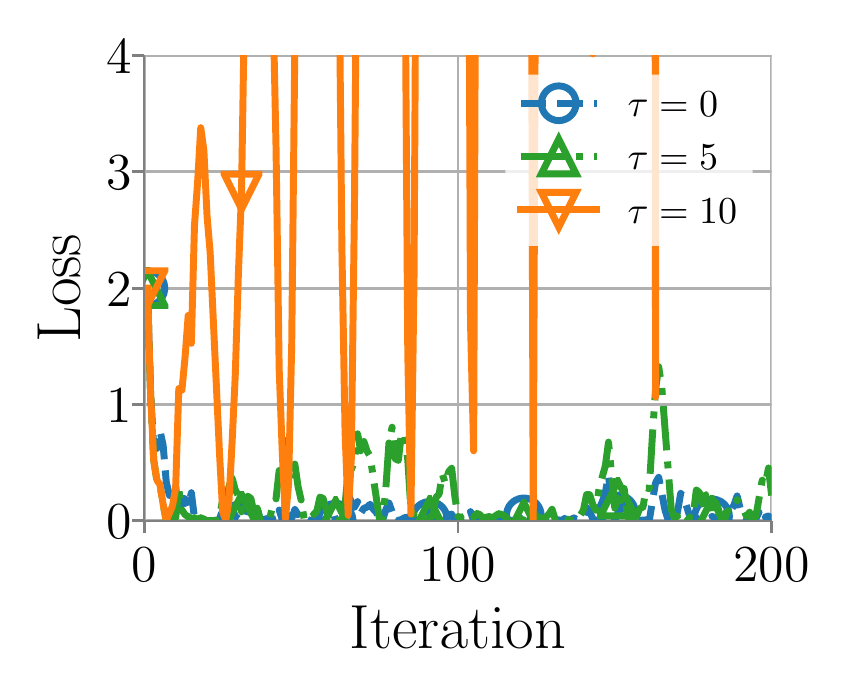} &
		\includegraphics[width=0.4\columnwidth]{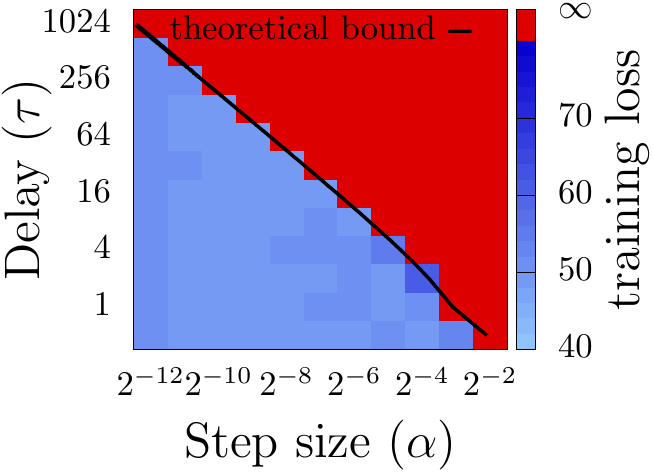} \\
		\else
		\includegraphics[width=0.47\columnwidth, trim=0 0.4cm 0 0cm]{figures/quadraticdivergence.pdf} &
		\hspace{-0.75em}\includegraphics[width=0.47\columnwidth]{figures/cpusmall.pdf} \\
		\fi
		(a) & (b)
	\end{tabular}
	\ifarxiv%
	\else%
	\vspace{-5mm}
	\fi
	\caption{(a) Increasing $\tau$ can cause the quadratic model to diverge even when $\alpha$ remains fixed. (b) Evaluation of pipeline-parallel SGD for linear regression on the \texttt{cpusmall} dataset
	running for $T = 10^6$ iterations.
    The heatmap reports losses as a function of the step size $\alpha$ and the delay $\tau$; red denotes divergence to $\infty$. The black curve shows the upper bound from Lemma~\ref{lemmaRootsGD} using the largest curvature of the objective in place of $\lambda$.
	}
	\ifarxiv%
	\else%
	\vspace{-3mm}
	\fi
	\label{fig:figCpuSmall}
\end{figure}
\textbf{The problem.}
We might hope that existing hyperparameters used for sequential SGD would ``just work'' for training in the asynchronous PP setting.
Unfortunately, when we try running naively with a standard step size scheme, asynchronous PP SGD can significantly underperform the synchronous baseline.
This happens because a large value of $\tau$ can cause SGD to diverge even when using a step size $\alpha$ for which the baseline synchronous algorithm converges.
This is shown in Figure~\ref{fig:figCpuSmall}(a), which simulates the quadratic model (5) with $\lambda = 1$, $\alpha = 0.2$, and $\eta_t \sim \mathcal{N}(0,1)$, for various values of $\tau$.
Notice that for $\tau = 10$, the trajectory diverges quickly.
In \Cref{app:dl_motivate}, we show that the same phenomenon can be observed for a Resnet50 network.

\textbf{The theory.}
The first question we want to ask is: \emph{when will asynchronous pipeline-parallel SGD be stable on the quadratic model?}
That is, for what values of the step size $\alpha$ will it be guaranteed that $w_t$ remains bounded, no matter what (bounded) noise signal $\eta_t$ we get from the gradient estimator?
To answer this question, notice that (\ref{eqnLTI}) is a linear system, which can be written in terms of a companion matrix that stores all the state of the system as
\[{\small
	\begin{bmatrix} w_{t+1} \\ w_t \\ \vdots \\ w_{t-\tau+1} \end{bmatrix}
	=
	\begin{bmatrix} 1 & 0 & \cdots & 0 & -\alpha \lambda \\ 1 & 0 & \cdots & 0 & 0 \\ \vdots & \vdots & \ddots & \vdots & \vdots \\ 0 & 0 & \cdots & 1 & 0 \end{bmatrix}
	\begin{bmatrix} w_t \\ w_{t-1} \\ \vdots \\ w_{t-\tau} \end{bmatrix}
	+
	\begin{bmatrix} \alpha \eta_t \\ 0 \\ \vdots \\ 0 \end{bmatrix}.
}\]
If we call this $(\tau + 1) \times (\tau + 1)$ companion matrix $C$, and call the vectorized version of $w$ with its history $W$, 
\begin{equation}
	\label{eqnCompanionMatrix}
	W_{t+1} = C W_t + \alpha \eta_t e_1,
\end{equation}
where $e_1$ is the vector $\begin{bmatrix} 1 & 0 & \cdots & 0 \end{bmatrix}^T$.
Linear systems of this type have solutions of the form
\[
	\textstyle
	w_t
	=
	\sum_{k=0}^{t-1} \alpha \eta_{t - k - 1} \sum_\omega \rho_{\omega}(k) \cdot \omega^k,
\]
where the sum here ranges over the eigenvalues $\omega$ of the companion matrix, and each $\rho_{\omega}$ is a polynomial of degree less than the algebraic multiplicity of the eigenvalue $\omega$.\footnote{To see why, consider the Jordan normal form of $C$, which will for each eigenvalue have a corresponding Jordan block of dimension equal to its algebraic multiplicity.}
Thus, the convergence of (\ref{eqnCompanionMatrix}) is determined entirely by $C$'s eigenvalues, 
and it will be stable when all these eigenvalues lie inside the unit disk in the complex plane.
$C$'s eigenvalues are the zeros of its characteristic polynomial
\begin{equation}
	p(\omega) = \omega^{\tau+1} - \omega^{\tau} + \alpha \lambda. \label{eqnCharPolyGD}
\end{equation}
So we want to find out for which values of $\alpha$ the roots of $p$ must all lie inside the unit disk.
\begin{lemma}
\label{lemmaRootsGD}The roots of the polynomial $p$ of (\ref{eqnCharPolyGD}) all lie inside the unit disk
if and only if the step size $\alpha$ is set such that
\ifarxiv%
\else%
\vspace{-2mm}%
\fi%
\[
	0 \le \alpha \le \frac{2}{\lambda} \cdot \sin\left( \frac{\pi}{4\tau + 2} \right) = O\left( \frac{1}{\lambda \tau} \right).
\]
\end{lemma}
\ifarxiv%
\else%
\vspace{-2mm}%
\fi%
This lemma gives us theoretical insight that backs up our empirical observations: when the delay is larger, the step size must be set smaller to prevent instability and divergence.
It also quantifies how much smaller, predicting that $\alpha$ should be set as $O(\tau^{-1})$.
In \Cref{fig:figCpuSmall}(b)
we validate that our theory not only applies to 1D optimization problems, but also can accurately describe what happens when we run pipeline-parallel SGD on a simple 12-dimensional linear regression problem using the cpusmall dataset~\cite{CC01a}; the algorithm diverges at precisely an $\alpha \propto \tau^{-1}$ slope, exactly what Lemma~\ref{lemmaRootsGD} predicts. In \Cref{app:ext_mom} we extend this to momentum SGD showing that the $O(\tau^{-1})$ threshold is general, which motivates our use of Technique 1 with learning algorithms other than SGD such as Adam.

\textbf{The technique.} To avoid the divergence we just characterized, the natural choice here is to divide the step size at each layer $i$ by its delay $\tau_i$.
However, this is (1) problematic because it leads to very small step sizes which slow convergence, and (2) unnecessary because it divides the step size by $\tau$ even for later epochs where the base step size has already become small, as is usually done in deep learning \cite{he2016deep,vaswani2017attention}.
This motivates us to develop a step size scheme that (1) behaves like the $O(\tau^{-1})$ scheme for early epochs, and (2) degrades back to the baseline learning rate scheme for later epochs.

\textbf{T1:}
\emph{Suppose that we are training a DNN.
In SGD step $k$, assign the following step size to layer $i$.
\begin{equation}
	\alpha_{k,i} = \frac{ \alpha_{k, \text{base}} }{ \tau_i^{p_k} }
	\hspace{1em}\text{where}\hspace{1em}
	p_k = 1 - \min\left(\frac{k}{K}, 1 \right).
\label{fig:heurStepSize}
\end{equation}
where $K$ is a hyperparameter representing a number of steps during which to adjust the learning rate, and $\alpha_{k, \text{base}}$ denotes the normal synchronous learning rate.
{ 
We suggest K to be one-quarter the length of the first phase of a fixed-step LR schedule (we use this for the ResNet model) or five times the linear warmup steps of a schedule with a linear warmup phase (we use this for the Transformer model).}
}

\subsection{Discrepancy correction (T2)}
\label{sec:delay_discrepancy}
In Section~\ref{sec:lr_rescheduling}, we analyzed a setting in which there was no delay discrepancy ($\taufwd = \taubkwd$).
In this subsection, we try to understand the effect of delay discrepancy, again using our quadratic model.
We then develop and evaluate a technique for ``correcting'' this discrepancy.

\begin{figure}[t]
	\ifarxiv
	\centering
	\vspace{-1.5em}
	\begin{tabular}{c c}
		\includegraphics[width=0.4\columnwidth]{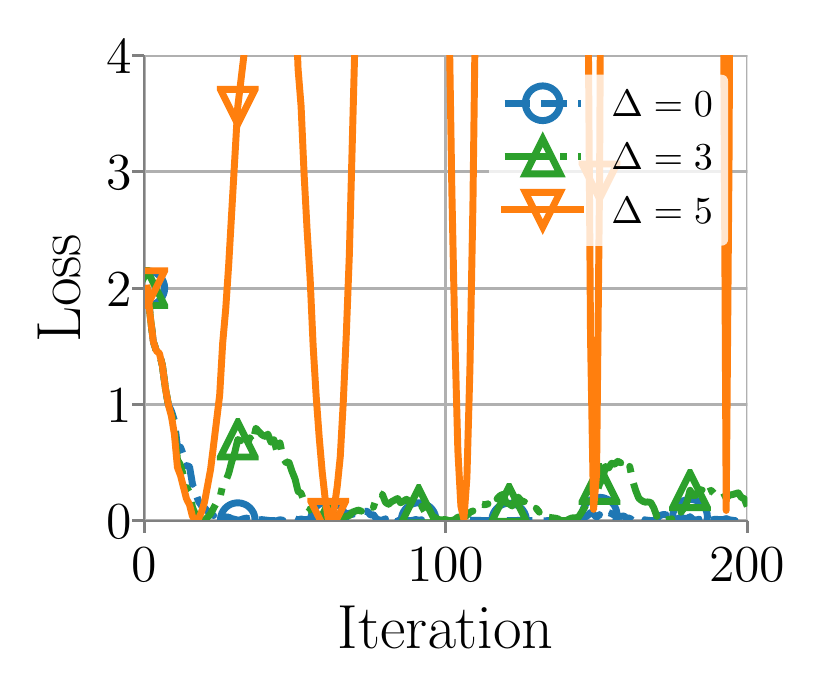} &
		\includegraphics[width=0.4\columnwidth]{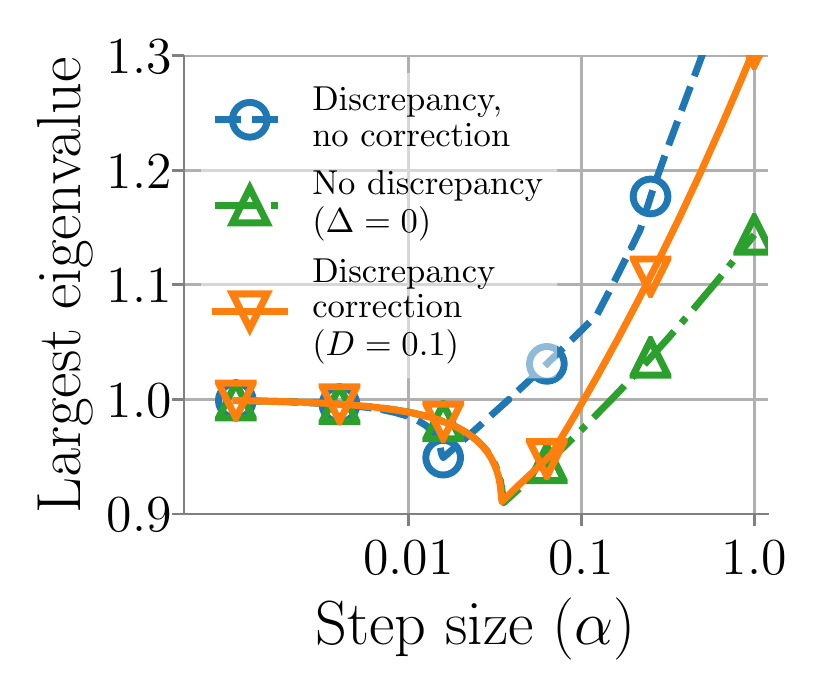} \\
		(a) & (b)
	\end{tabular}
	\else
	\vspace{-1mm}

	\hspace{-4pt}\begin{tabular}{c c}
		\includegraphics[width=0.475\columnwidth]{figures/quadraticdiscrepancy.pdf} &
		\includegraphics[width=0.475\columnwidth]{figures/momentumcorr.pdf} \\
	\vspace{-1mm}
		(a) & (b)
	\end{tabular}
	\vspace{-3mm}
	\fi
	\caption{(a) Increasing $\Delta$, the gradient sensitivity to delay discrepancy, can cause the quadratic model to diverge even when $\alpha$ and $\tau$ remain fixed, using $\taufwd = 10$, $\taubkwd = 6$, and $\lambda = 1$. (b) Effect of discrepancy correction on the quadratic model. 
	Forward-backward delay discrepancy (blue) increases the largest magnitude eigenvalue of the companion matrix with $\Delta = 5$, and $\tau$, $\lambda$ same as in (a). 
	Discrepancy correction with $\operatorname{D} = 0.1$ (orange) reduces the largest magnitude eigenvalue; this eigenvalue is closer to that attained without delay discrepancy (green).
	}
		\vspace{-4mm}
	\label{fig:figDelayDiscrepancy}
\end{figure}

\textbf{The problem.}
To model delay discrepancy, we now assume gradient samples of the form
\[
	\nabla f_t(u_{\text{fwd},t}, u_{\text{bkwd},t}) = (\lambda + \Delta) \cdot w_{t-\tau_{\text{fwd}}} - \Delta \cdot w_{t-\tau_{\text{bkwd}}} - \eta_t
\]
where $\tau_{\text{fwd}} > \tau_{\text{bkwd}}$ are two different delays, and $\Delta$ is a constant that measures the sensitivity of the gradients to discrepancy.
We can think of this as the natural first-order (linear) approximation of $\nabla f_t$ in the neighborhood of a stationary point---it models any affine function of $u_{\text{fwd},t}$ and $u_{\text{bkwd},t}$ that is consistent with the curvature $\lambda$ when $u_{\text{fwd},t} = u_{\text{bkwd},t}$. 
If $\Delta = 0$, we recover a model of our original zero-discrepancy setting, whereas for large-magnitude values of $\Delta$, even a small delay discrepancy could be amplified to have a large effect on the gradient samples.

Delay discrepancy is problematic because it can amplify the divergence effect observed in Section~\ref{sec:lr_rescheduling}.
To illustrate, Figure~\ref{fig:figDelayDiscrepancy}(a) shows on the quadratic model (with $\taufwd = 10$, $\taubkwd = 6$, $\lambda = 1$, and $\eta_t \sim \mathcal{N}(0,1)$) that a nonzero value of $\Delta$ can cause divergence even for a value of $\alpha$ and $\tau$ where with $\Delta = 0$ (i.e.\ running PipeDream-style with no discrepancy) the trajectory would converge.
In \Cref{app:dl_motivate}, we illustrate that, just as was the case for the divergence phenomenon of Section~\ref{sec:lr_rescheduling}, on ResNet50 asynchronous SGD with a large enough $\Delta$ will diverge even for values of $\alpha$ and $\tau$ for which PipeDream-style SGD converged.
We seek to understand this phenomenon theoretically and to develop a technique to limit its effect.

\textbf{The theory.}
With our new discrepancy-dependent samples, pipeline-parallel SGD on our model has the update step
\begin{equation}
	w_{t+1} 
	=
	w_{t} - \alpha (\lambda + \Delta) w_{t-\tau_{\text{fwd}}} + \alpha \Delta w_{t-\tau_{\text{bkwd}}} + \alpha \eta_t.
	\label{eqnUpdateDD}
\end{equation}
As before, we can analyze this for stability by finding the value of $\alpha$ for which the roots of its characteristic polynomial lie inside the unit disk.
\begin{lemma}
\label{lemmaRootsDD}For any $\Delta > 0$, there exists an $\alpha > 0$ with
\[
    \alpha
    \le
    \min\left(
        \frac{2}{\Delta \cdot \left(\taufwd - \taubkwd \right)},
        \frac{2}{\lambda} \cdot \sin\left( \frac{\pi}{4\taufwd + 2} \right)
    \right)
\]
such that at least one of the roots of the characteristic polynomial of (\ref{eqnUpdateDD}) is outside the interior of the unit disk (that is, the disrepancy-dependent model updates will be unstable).
\end{lemma}
This lemma shows two important things: first, that the maximal stable step size is still inversely proportional to the delay, even with delay discrepancy; second, that for large values of $\Delta$, in which the delay discrepancy has substasntial effect on the gradient, the largest stable $\alpha$ becomes smaller (although still inversely proportional to $\tau$).
This models the behavior illustrated in Figure~\ref{fig:figDelayDiscrepancy}(a) where adding delay discrepancy exacerbates the divergence phenomenon.

\textbf{The technique.}
As shown, delay discrepancy between the forward and backward passes can exacerbate the problem of divergence.
If we could just compute $\nabla f_t(u_{\text{fwd},t}, u_{\text{fwd},t})$ directly, then this mismatch would not be a problem.
Unfortunately, in our asynchronous PP setting we cannot compute this, as we no longer have $u_{\text{fwd},t}$ in memory by the time the backward pass comes around .
To keep $u_{\text{fwd},t}$ stored in memory is possible, but undesirable as it would greatly increase memory usage (as in PipeDream).
Instead, we decrease the gap between $u_{\text{fwd},t}$ and $u_{\text{bkwd},t}$ by \emph{approximating} $u_{\text{fwd},t}$ 
without storing the full history of model weight values after $u_{\text{fwd},t}$,
using a bit of extra memory to hold an approximation of the velocity of the weights.

\textbf{T2:}
\emph{Instead of the assignment of $u_{\text{bkwd}}$ from Section~\ref{sec:lr_rescheduling}, set
\vspace{-3mm}
\[
	\left( u_{\text{bkwd},t} \right)_i
	=
	\left( w_{t - \tau_{\text{bkwd},i}} \right)_i
	-
	\left( \tau_{\text{fwd},i} - \tau_{\text{bkwd},i} \right) \delta_{t, i},
\]
where $\delta_{t, i}$ is a newly added accumulator that estimates the amount that $w_{i}$ is changing over time.
It is kept up to date by the update step
$\delta_{t + 1, i} = \gamma_i \cdot \delta_{t, i} + (1 - \gamma_i) \cdot \left( w_{t+1,i} - w_{t,i} \right)$,
where $\gamma_i$ is a decay rate parameter, assigned per-stage to
$\gamma_i = \operatorname{D}^{1/(\tau_{\text{fwd},i} - \tau_{\text{bkwd},i})}$,
where $\operatorname{D}$ is a tunable global hyperparameter.}

Essentially, this technique adjusts the value of the weights used in the backward pass by extrapolating what the weights were during the forward pass based on the recent average trajectory of the weights.
Applying T2 on the quadratic model also results in an update step that can be modeled with a companion matrix; we analyzed this system---just as before---by considering that companion matrix's eigenvalues.
Doing this, we observed that 
T2 seems to increase the allowable range of $\alpha$ for which the quadratic model is stable.
This is illustrated in Figure~\ref{fig:figDelayDiscrepancy}(b).

\vspace{-2mm}
\section{Experiments}
\label{sec:experiments}

We evaluate PipeMare
on two standard image recognition tasks and neural machine translation
tasks.  Our evaluation supports the following two main claims:
\begin{itemize}[nosep, itemindent=24pt,leftmargin=0pt]
\item \emph{PipeMare enables more efficient end-to-end training.}  We show that across two image recognition and two neural machine translation tasks, PipeMare can attain up to $4.6\times$ higher pipeline utilization over the synchronous GPipe; we also show that PipeMare can attain a final model quality that PipeDream cannot reach, while using up to $2\times$ less weight and optimizer memory.
\item \emph{PipeMare achieves final model qualities similar to those attained by synchronous training.} \update{We show that PipeMare can achieve
a final model accuracy within $0.1\%$ of synchronous baselines on image recognition tasks and match  the BLEU score of synchronous baselines on neural machine translation tasks.} 
\end{itemize}

\textbf{Warmup Epochs (W).}\hspace{1em}
In some cases, the statistical-efficiency hardware-efficiency tradeoff PipeMare presents is too coarse-grained.
Here, we use the standard technique of running a number of \emph{warmup epochs} of the baseline method before switching to using PipeMare.
This is another way to trade off hardware efficiency (since the warmup epochs are less efficient) for statistical efficiency.
Concretely, we initialize with $E_w$ epochs of synchronous (GPipe-style) pipeline-parallel SGD using the standard learning rate.
We call this modified method \textbf{PipeMareW}.

\newcolumntype{H}{>{\setbox0=\hbox\bgroup}c<{\egroup}@{}}

\begin{table}[t]
	\centering
		\setlength{\tabcolsep}{3.5pt}
	\ifarxiv
	\else
	\small
	\fi
	\caption{Comparison of statistical efficiency (metric), pipeline utilization, and weight+optimizer memory of PipeMare and baselines. Here we use top-1 accuracy or BLEU score as the metrics for CIFAR10/ImageNet and IWSLT/WMT respectively.
	}
	\ifarxiv
	\begin{tabular}{l  l  c  H  H  H  c  c  c  c}
	\else
	\begin{tabular}{l  l  r  H  H  H  r  r  r  r}
	\fi
		\toprule
		\ifarxiv
		Dataset & Method & Best metric & Target metric & \makecell{Speedup\\to Target} & \makecell{Epochs\\to Target} & \makecell{Pipe. Util.} & \makecell{Weight+optimizer\\Memory} \\
		\else
		Dataset & Method & Metric & Target metric & Speedup to Target & Epochs to Target & \makecell{Pipe. Util.} & Memory \\
		\fi
		\midrule
		\multirow{3}{*}{CIFAR10} 
                              & PipeDream 	& 94.8 			& \multirow{3}{*}{94.0}	& \textbf{3.3X} & \textbf{82} 	& \textbf{100\%} & 2.70X \\
                               & GPipe     	& \textbf{95.0} &  						& 1.0X  		& 83 	  		& 23\%  		& \textbf{1X} (270MB) \\
	                            & PipeMare  & \textbf{95.0} &                       & \textbf{3.3X} & \textbf{82} 	& \textbf{100\%} & 1.33X \\
		\midrule
    \multirow{4}{*}{ImageNet} 
    	                        & PipeDream & 74.7 		& \multirow{3}{*}{75.4}	& --- 	& --- 	& \textbf{100\%}	& 1.61X \\
    	                        & GPipe     & \textbf{76.0} &  						& 1.0X  		& \textbf{70} 	   	& 13\%  		& \textbf{1X} (293MB) \\
	                            & PipeMare  & 75.5 &                       			&\textbf{2.5X} 	& 94 		& \textbf{100\%}   & 1.33X \\
	                            & PipeMareW  & 75.9 &                       			& 1.9X 	& 91 		& 33\%   & 1.33X \\
	  \midrule
    \multirow{4}{*}{IWSLT14} 
                              & PipeDream 	& 0.0  &  \multirow{3}{*}{34.1}	& ---  			& --- 		    & \textbf{100\%}           & 2.06X \\
                              & GPipe     	& \textbf{34.5} & 				& 1.0X  	    & \textbf{30} 	& 17\%  		          & \textbf{1X} (0.65GB) \\
  	                          & PipeMare  	& 34.1 &              	& \textbf{1.6X} & 60 			& \textbf{100\%}  		 		  & 1.25X\\
  	                          & PipeMareW  	& \textbf{34.5} &              	& \textbf{1.7X} & 35 			& 55\%  		 		  & 1.25X\\
	  \midrule
    \multirow{4}{*}{WMT17}    
        	                    & PipeDream  	& 0.0 & \multirow{3}{*}{27.4}	& ---  	& --- 		& \textbf{100\%}		& 2.39X \\
        	                    & GPipe     	& 27.5 & 					& 1.0X  		& \textbf{50} 	& 56\%  		    & \textbf{1X} (1.01GB) \\
  	                            & PipeMare  	& 27.0 &                       	    & NA & 68 			& \textbf{100\%} 	& 1.25X \\  	              
  	                            & PipeMareW  	& \textbf{27.8} &                       	    & \textbf{2.6X} & 54 			& 96\% 	& 1.25X \\
		\bottomrule
	\end{tabular}
	\label{tab:tradeoffs_macro}
\vspace{-5mm}
\end{table}

\subsection{Experimental Setting}

We overview the details of our experimental setup and refer the reader to \Cref{app:exp_app}
for the exact details.

\begin{figure*}[t]
\includegraphics[width=\linewidth]{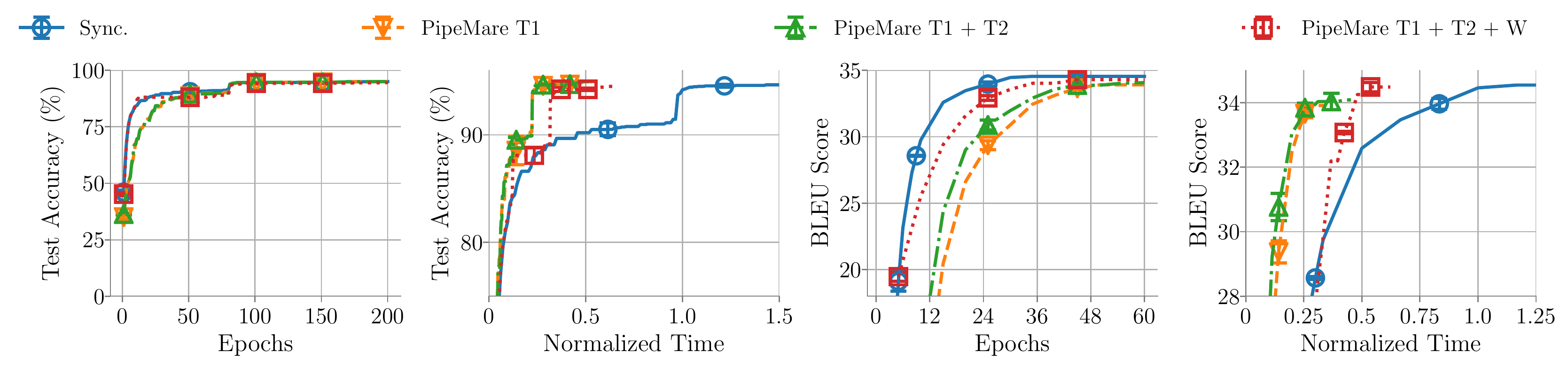}
\vspace{-8mm}
\caption{Tradeoffs when incrementally combining PipeMare techniques (T1, T2, and W) with ResNet50 on Cifar10 (leftmost two figures) and Transformer on IWSLT (rightmost two figures). We set the number of pipeline stages in the ResNet50 and the 12-layer Transformer models to $107$ stages and $93$ which is the number of pipeline stages when each model weight is treated at its own stage (or the finest granularity). Normalized time is computed using pipeline utilization and number of epochs, providing a proxy for the idealized time on an accelerator. Runs are stable across seeds indicated via the (negligible) error bars in each plot.}
\vspace{-4mm}
\label{fig:trade-off-1x}
\end{figure*}

\textbf{Setup.}\hspace{1em} 
Since the purpose of this paper is to determine if asynchronous PP
is feasible statistically (and, if successful, influence future hardware accelerator designs to reap the hardware benefits), we built a custom optimizer in PyTorch that simulates 
the exact asynchronous updates (via a queue of stale weights) that would occur
during fine-grained PP training. Using this, we report the pipeline utilization, weight and optimizer memory, and best accuracy (or 
BLEU score) on each benchmark. The pipeline utilization and memory we report are calculated
using the formulas we present in \Cref{tab:taus}. We report the averaged model accuracy from runs with three different random seeds. 

\textbf{Benchmarks.}\hspace{1em}We benchmark a ResNet50 model~\cite{he2016deep} for image classification and the 12-layer Transformer model~\cite{vaswani2017attention} 
to benchmark neural machine translation: each represents standard benchmarks in 
their respective domains \cite{mlperf}. We use the standard CIFAR10 and ImageNet datasets for image classification, and popular IWSLT14 German-to-English and WMT17 English-to-German dataset for neural machine translation. 
For image classification, we use test set accuracy as the model accuracy metric while in translation tasks we use test BLEU score. We compare PipeMare to two synchronous (baseline) PP
training methods: GPipe and PipeDream. We report in detail on the two non-standard hyperparameters we had to select next (microbatch size and number of pipeline stages). For all other hyperparameters we use standard, publicly available hyperparameters (see \Cref{app:model_dataset}) for each of these two popular models.

\textbf{Microbatch Size.}\hspace{1em}For microbatch size ($M$) we always select a value that is as small as possible. This has two main benefits: (1) it saves activation memory 
 and (2) it results in less gradient delay $\tau_{fwd}$ given a fixed number of pipeline stages (more microbatches per minibatch). In more detail, we choose $M=8$ for ResNet50 on CIFAR10 and $M=16$ for ResNet50 on ImageNet as smaller $M$, in both cases, cause problems in batch normalization~\cite{yuxin2018groupnorm} layers. For Transformer on IWSLT14 we choose the maximum tokens per microbatch to be 245 because this is the number of tokens in the longest sentence (and therefore is the smallest size that does not lose information within a sentence). On WMT17, we used a maximum tokens per microbatch of 1792 for both PipeDream and PipeMare. We choose this as it was the smallest size that provided reasonable hardware utilization on Nvidia V100 GPUs---enabling results to be produced within a reasonable timeframe and budget. To be fair, for GPipe on WMT17 we calculated their pipeline utilization using a maximum tokens per microbatch of 251 (longest sentence in WMT17)---maximizing their utilization and without impacting statistical efficiency.

\textbf{Pipeline Stages.}\hspace{1em}
 \ifarxiv
 \update{To partition the model into pipeline stages during training, we traverse model weights according to their topological order in the computation graph, always treating the weight and bias in the same layer as a single model weight (i.e. always in the same pipeline stage). Next, we divide these model weights evenly into P stages to split model weights across pipeline stages. This represents a very fine granularity of pipeline
 parallelism which is a difficult one to train. Specifically, with ResNet50 we use 107 stages and with 12-layer Transformer we use 91 or 93 stages\footnote{Transformer model for WMT17 employs shared embedding between encoder, decoder and projection while IWSLT14 has independent embeddings.}; these are the largest number of stages with at least one model weight assigned to each pipeline stage.}
 \else
 To select the number of pipeline stages during training, we traverse model weights according to their topological order in the graph, always treating weight and bias as a single weight (i.e. always in the same pipeline stage). Next, we divide these weights evenly into P stages to split model weights across pipeline stages. This represents a very fine granularity of pipeline
 parallelism to extract on hardware (and a difficult one to train). Specifically, with ResNet50 we use 107 stages and with 12-layer Transformer we use 93 stages.
 \fi
\vspace{-4mm}
\paragraph{Implementation Details} All experiments are run using a simulator we built in PyTorch and on AWS p3.2xlarge instances (Nvidia V100 GPUs). Our simulator maintains a queue of weight values over time to simulate the exact delay one would see when running fine-grained PP.

\subsection{End-to-End Comparison}
\label{sec:empirical_eval}
We compare the asynchronous PipeMare training method to the synchronous GPipe and PipeDream methods on both image classification and machine translation tasks. In Table~\ref{tab:tradeoffs_macro} we show that on both of these tasks PipeMare achieves higher pipeline utilization while achieving comparable final model qualities---\update{the greatest difference being a 0.1\% top-1 accuracy difference on ImageNet.}

\textbf{Image classification tasks}\hspace{1em}As shown in Table~\ref{tab:tradeoffs_macro}, on both the CIFAR10 and ImageNet dataset, PipeMare can respectively achieve \update{$4.3\times$} higher pipeline utilization than GPipe. Note that on CIFAR10, PipeMare attains a perfect pipeline utilization of 100\% because we do not need any warmup epochs here.  Though PipeDream attains the same pipeline utilization as PipeMare here, PipeDream requires $2.7\times$ more weight and optimizer memory (see Table~\ref{tab:tradeoffs_macro}). PipeMareW has 0.1\% accuracy gap with GPipe on ImageNet while achieving 2.5x higher pipeline utilization even though it uses 30 synchronous warmup epochs. 

\textbf{Neural machine translation tasks}\hspace{1em}As demonstrated in Table~\ref{tab:tradeoffs_macro}, \update{PipeMareW can achieve 3.2$\times$ and 1.7$\times$ higher pipeline utilization than GPipe on the IWSLT14 and WMT17 datasets.  When comparing PipeMare(W) to the PipeDream approach, we observe that PipeDream fails to train Transformer even though it uses $>$2$\times$ more weight and optimizer memory than PipeMare. On the other hand GPipe trains the model fine but sacrifices either pipeline utilization or activation memory to maintain its statistical efficiency. 
Because we use PipeMareW (warmup epochs) on both the IWSLT14 (10 warmup epochs) and WMT17 (4 warmup epoch) experiments respectively the amortized pipeline utilization of PipeMare is smaller than $1$}. Still, by combining our techniques we improve the pipeline utilization and memory usage, when compared to previous PP techniques, with no loss in statistical performance. 

\begin{table}[t]
	\centering
	\ifarxiv
	\else
	\small
	\fi
	\setlength{\tabcolsep}{2.7pt}
	\caption{Ablation study of PipeMare. We show the impact of the learning rate rescheduling (T1), discrepancy correction (T2), and warmup epochs (W) on metrics (test accuracy or BLEU score) of interest. Note that warmup epochs were not necessary on the CIFAR10 dataset to recover the performance attained by GPipe.
}
	\ifarxiv
	\begin{tabular}{l l  c H H H c c c}
	\else
	\begin{tabular}{l l  r  H  H  H  r  r  r}
	\fi
		\toprule
		\ifarxiv
		Dataset & Method & Best metric & Target metric & \makecell{Speedup\\to Target} & \makecell{Epochs\\to Target} & \makecell{Pipe. Util.} & \makecell{Weight+optimizer\\Memory} \\
		\else
		Dataset & Method & Metric & Target metric & Speedup to Target & Epochs to Target & \makecell{Pipe. Util.} & Memory \\
		\fi
		\midrule
    \multirow{3}{*}{CIFAR10}
                             & T1 Only            & \textbf{95.0}  & \multirow{3}{*}{94.0}   & \textbf{3.3X}   		& 83 			& \textbf{100\%}   		 & \textbf{1X} (270MB) \\
                             & T2 Only            & 94.5  &                        			 & 3.2X  		& 86 		 	& \textbf{100\%}   		 & 1.33X \\
                             & T1+T2              & \textbf{95.0}  &                		 & \textbf{3.3X} & \textbf{82} 	& \textbf{100\%}  & 1.33X \\
	 \midrule
   \multirow{4}{*}{IWSLT14} 
                             & T1 Only          & 34.1  & \multirow{4}{*}{34.1}   	& 1.6X  		& 60 			& \textbf{100\%}   		& \textbf{1X} (0.65GB) \\
                             & T2 Only          & 0.0   &                    		& -  			& -  			& \textbf{100\%}   		& 1.25X \\
  	 						 & T1 + T2 Only 	& 34.1  &                        	& 1.6X 			& 60 			& \textbf{100\%} & 1.25X \\	
  	 						 & T1 + T2 + W     & \textbf{34.5}  &      		  	& \textbf{1.7X} & \textbf{35}   & 55\%   		& 1.25X \\
	 \bottomrule
	\end{tabular}
	\label{tab:tradeoffs_micro_best_metric}
	\vspace{-3mm}
\end{table}

\subsection{Ablation study}
\label{sec:ablation}
To understand the contribution of each technique to the performance of PipeMare, we perform ablation studies on PipeMare with respect to memory, pipeline utilization, and model quality. We show that each technique is necessary for PipeMare to outperform synchronous techniques from both a hardware and statistical efficiency perspective. This study is summarized in \Cref{fig:trade-off-1x,tab:tradeoffs_micro_best_metric}
\vspace{-2.5mm}
\paragraph{Learning rate rescheduling (T1)}
The asynchronous PP training method with only learning rate rescheduling fully utilizes the compute power by avoiding both bubbles in the execution pipeline and additional weight memory. Therefore it achieves optimal hardware efficiency when compared to any other approach. 
In \Cref{tab:tradeoffs_micro_best_metric} we show that this alone can achieve a test accuracy of $95.0\%$ and a test BLEU score of $34.1$, both of which are competitive to the baseline $95.0\%$ accuracy and $34.5$ BLEU score of synchronous methods. In terms of pipeline utilization, learning rate rescheduling alone achieves $7.6\times$ improvement over GPipe---indicating its important role in improving statistical efficiency as well as hardware efficiency.  
 For ResNet50 on CIFAR10, the test accuracy of asynchronous training with learning rate rescheduling matches that of synchronous training while asynchronous training without it diverges---emphasizing the importance of T1 during synchronous training. For the Transformer model, T1 takes about twice as many epochs of synchronous training to reach BLEU score $34.1$ while asynchronous training without T1 achieves a test BLEU score $\leq1$. 
\vspace{-1.5mm}
\paragraph{Discrepancy Correction (T2)}
As shown in \cref{tab:tradeoffs_micro_best_metric}, discrepancy correction in isolation achieves a test accuracy of $94.5\%$ for ResNet 50 and a jarring $0.0$ test BLEU score on the Transformer model. The poor Transformer model training is fixed by combining discrepancy correction with learning rate rescheduling,
though the final BLEU score achieved is the same as in the learning rate scheduling only setting ($34.1$). Discrepancy correction with learning rate
rescheduling shines on the convergence speed of both models, especially Transformer model on IWSLT14, as is seen in \Cref{fig:trade-off-1x} and ~\Cref{fig:trade-off-2x}. This of course comes at the cost of using more weight memory. We find this cost to be minimal ($33\%$ more for SGD with momentum and $25\%$ more for ADAM) when compared
to the final model quality improvements from using this technique in conjunction with learning rate rescheduling. To further validate the efficacy of discrepancy correction, in~\Cref{app:ablation}, we show that on a ResNet 152 model with 150 stages discrepancy correction is necessary to prevent divergence and match the model accuracy attained by synchronous training. 
\vspace{-1.5mm}
\paragraph{Warmup Epochs (PipeMareW)}
As shown in Table~\ref{tab:tradeoffs_micro_best_metric} learning rate rescheduling and discrepancy correction leave a noticeable BLUE score gap ($0.4$) for the Transformer model running on the IWSLT14 dataset. To close this gap PipeMareW adds 10 synchronous warmup epochs. As shown in \cref{tab:tradeoffs_micro_best_metric}, the best BLEU score attained by asynchronous training is boosted from $34.1$ to $34.5$. This of course comes at the cost of decreasing the overall pipeline utilization from 100\% to 55\%, which still enables PipeMare to outperform its baselines (higher pipeline utilization than GPipe and PipeDream does not converge). 

\section{Conclusion}
\label{sec:conclusion}

In this paper, we presented PipeMare, a system for asynchronous pipeline-parallel training of DNN models.
PipeMare uses a bubble-free PP hardware model along with two theoretically motivated techniques (learning rate rescheduling and discrepancy correction) which help improve statistical efficiency.
Experimentally, we showed PipeMare has better hardware efficiency (pipeline utilization and memory) than competing algorithms. We hope that this will make PipeMare a promising candidate algorithm for the new generation of hardware chips designed for training DNNs.
 
\bibliographystyle{icml2020}
\bibliography{references}

\begin{thebibliography}{23}
\providecommand{\natexlab}[1]{#1}
\providecommand{\url}[1]{\texttt{#1}}
\expandafter\ifx\csname urlstyle\endcsname\relax
  \providecommand{\doi}[1]{doi: #1}\else
  \providecommand{\doi}{doi: \begingroup \urlstyle{rm}\Url}\fi

\bibitem[mlp(2019)]{mlperf}
Fair and useful benchmarks for measuring training and inference performance of
  ml hardware, software, and services.
\newblock 2019.
\newblock URL \url{https://mlperf.org/}.

\bibitem[Chang \& Lin(2011)Chang and Lin]{CC01a}
Chang, C.-C. and Lin, C.-J.
\newblock {LIBSVM}: A library for support vector machines.
\newblock \emph{ACM Transactions on Intelligent Systems and Technology},
  2:\penalty0 27:1--27:27, 2011.
\newblock Software available at \url{http://www.csie.ntu.edu.tw/~cjlin/libsvm}.

\bibitem[Chen et~al.(2016{\natexlab{a}})Chen, Xu, Zhang, and
  Guestrin]{chen2016training}
Chen, T., Xu, B., Zhang, C., and Guestrin, C.
\newblock Training deep nets with sublinear memory cost.
\newblock \emph{arXiv preprint arXiv:1604.06174}, 2016{\natexlab{a}}.

\bibitem[Chen et~al.(2016{\natexlab{b}})Chen, Xu, Zhang, and
  Guestrin]{recompute}
Chen, T., Xu, B., Zhang, C., and Guestrin, C.
\newblock Training deep nets with sublinear memory cost.
\newblock \emph{CoRR}, abs/1604.06174, 2016{\natexlab{b}}.
\newblock URL \url{http://arxiv.org/abs/1604.06174}.

\bibitem[De~Sa et~al.(2015)De~Sa, Zhang, Olukotun, and R{\'e}]{de2015taming}
De~Sa, C.~M., Zhang, C., Olukotun, K., and R{\'e}, C.
\newblock Taming the wild: A unified analysis of hogwild-style algorithms.
\newblock In \emph{Advances in neural information processing systems}, pp.\
  2674--2682, 2015.

\bibitem[Feldman(2019)]{cerebras}
Feldman, A.
\newblock Cerebras wafer scale engine: An introduction.
\newblock 2019.
\newblock URL
  \url{https://www.cerebras.net/wp-content/uploads/2019/08/Cerebras-Wafer-Scale-Engine-Whitepaper.pdf}.

\bibitem[Harlap et~al.(2018)Harlap, Narayanan, Phanishayee, Seshadri, Devanur,
  Ganger, and Gibbons]{harlap2018pipedream}
Harlap, A., Narayanan, D., Phanishayee, A., Seshadri, V., Devanur, N., Ganger,
  G., and Gibbons, P.
\newblock Pipe{D}ream: Fast and efficient pipeline parallel {DNN} training.
\newblock \emph{arXiv preprint arXiv:1806.03377}, 2018.

\bibitem[He et~al.(2016)He, Zhang, Ren, and Sun]{he2016deep}
He, K., Zhang, X., Ren, S., and Sun, J.
\newblock Deep residual learning for image recognition.
\newblock In \emph{Proceedings of the IEEE conference on computer vision and
  pattern recognition}, pp.\  770--778, 2016.

\bibitem[Huang et~al.(2018)Huang, Cheng, Chen, Lee, Ngiam, Le, and
  Chen]{huang2018gpipe}
Huang, Y., Cheng, Y., Chen, D., Lee, H., Ngiam, J., Le, Q.~V., and Chen, Z.
\newblock Gpipe: Efficient training of giant neural networks using pipeline
  parallelism.
\newblock \emph{arXiv preprint arXiv:1811.06965}, 2018.

\bibitem[Jouppi et~al.(2017)Jouppi, Young, Patil, Patterson, Agrawal, Bajwa,
  Bates, Bhatia, Boden, Borchers, et~al.]{jouppi2017datacenter}
Jouppi, N.~P., Young, C., Patil, N., Patterson, D., Agrawal, G., Bajwa, R.,
  Bates, S., Bhatia, S., Boden, N., Borchers, A., et~al.
\newblock In-datacenter performance analysis of a tensor processing unit.
\newblock In \emph{2017 ACM/IEEE 44th Annual International Symposium on
  Computer Architecture (ISCA)}, pp.\  1--12. IEEE, 2017.

\bibitem[Kurth et~al.(2017)Kurth, Zhang, Satish, Racah, Mitliagkas, Patwary,
  Malas, Sundaram, Bhimji, Smorkalov, et~al.]{kurth2017deep}
Kurth, T., Zhang, J., Satish, N., Racah, E., Mitliagkas, I., Patwary, M. M.~A.,
  Malas, T., Sundaram, N., Bhimji, W., Smorkalov, M., et~al.
\newblock Deep learning at 15pf: supervised and semi-supervised classification
  for scientific data.
\newblock In \emph{Proceedings of the International Conference for High
  Performance Computing, Networking, Storage and Analysis}, pp.\ ~7. ACM, 2017.

\bibitem[Liu et~al.(2019)Liu, Ott, Goyal, Du, Joshi, Chen, Levy, Lewis,
  Zettlemoyer, and Stoyanov]{roberta}
Liu, Y., Ott, M., Goyal, N., Du, J., Joshi, M., Chen, D., Levy, O., Lewis, M.,
  Zettlemoyer, L., and Stoyanov, V.
\newblock Roberta: {A} robustly optimized {BERT} pretraining approach.
\newblock \emph{CoRR}, abs/1907.11692, 2019.
\newblock URL \url{http://arxiv.org/abs/1907.11692}.

\bibitem[Mitliagkas et~al.(2016)Mitliagkas, Zhang, Hadjis, and
  R{\'e}]{mitliagkas2016asynchrony}
Mitliagkas, I., Zhang, C., Hadjis, S., and R{\'e}, C.
\newblock Asynchrony begets momentum, with an application to deep learning.
\newblock In \emph{2016 54th Annual Allerton Conference on Communication,
  Control, and Computing (Allerton)}, pp.\  997--1004. IEEE, 2016.

\bibitem[Recht et~al.(2011{\natexlab{a}})Recht, Re, Wright, and
  Niu]{recht2011hogwild}
Recht, B., Re, C., Wright, S., and Niu, F.
\newblock Hogwild: A lock-free approach to parallelizing stochastic gradient
  descent.
\newblock In \emph{Advances in neural information processing systems}, pp.\
  693--701, 2011{\natexlab{a}}.

\bibitem[Recht et~al.(2011{\natexlab{b}})Recht, R{\'{e}}, Wright, and
  Niu]{hogwild}
Recht, B., R{\'{e}}, C., Wright, S.~J., and Niu, F.
\newblock Hogwild: {A} lock-free approach to parallelizing stochastic gradient
  descent.
\newblock In \emph{Advances in Neural Information Processing Systems 24: 25th
  Annual Conference on Neural Information Processing Systems 2011. Proceedings
  of a meeting held 12-14 December 2011, Granada, Spain.}, pp.\  693--701,
  2011{\natexlab{b}}.
\newblock URL
  \url{http://papers.nips.cc/paper/4390-hogwild-a-lock-free-approach-to-parallelizing-stochastic-gradient-descent}.

\bibitem[Sutskever et~al.(2013)Sutskever, Martens, Dahl, and
  Hinton]{sutskever2013importance}
Sutskever, I., Martens, J., Dahl, G., and Hinton, G.
\newblock On the importance of initialization and momentum in deep learning.
\newblock In \emph{International conference on machine learning}, pp.\
  1139--1147, 2013.

\bibitem[Szegedy et~al.(2016)Szegedy, Ioffe, and Vanhoucke]{inception}
Szegedy, C., Ioffe, S., and Vanhoucke, V.
\newblock Inception-v4, inception-resnet and the impact of residual connections
  on learning.
\newblock \emph{CoRR}, abs/1602.07261, 2016.
\newblock URL \url{http://arxiv.org/abs/1602.07261}.

\bibitem[Vaswani et~al.(2017)Vaswani, Shazeer, Parmar, Uszkoreit, Jones, Gomez,
  Kaiser, and Polosukhin]{vaswani2017attention}
Vaswani, A., Shazeer, N., Parmar, N., Uszkoreit, J., Jones, L., Gomez, A.~N.,
  Kaiser, {\L}., and Polosukhin, I.
\newblock Attention is all you need.
\newblock In \emph{Advances in neural information processing systems}, pp.\
  5998--6008, 2017.

\bibitem[Ward-Foxton(2019{\natexlab{a}})]{graphcore}
Ward-Foxton, S.
\newblock Graphcore ceo touts 'most complex processor' ever.
\newblock \emph{EE Times}, 2019{\natexlab{a}}.
\newblock URL \url{{https://www.eetimes.com}}.

\bibitem[Ward-Foxton(2019{\natexlab{b}})]{habana}
Ward-Foxton, S.
\newblock Habana debuts record-breaking ai training chip.
\newblock \emph{EE Times}, 2019{\natexlab{b}}.
\newblock URL \url{{https://www.eetimes.com}}.

\bibitem[Wilson et~al.(2017)Wilson, Roelofs, Stern, Srebro, and
  Recht]{wilson2017marginal}
Wilson, A.~C., Roelofs, R., Stern, M., Srebro, N., and Recht, B.
\newblock The marginal value of adaptive gradient methods in machine learning.
\newblock In \emph{Advances in Neural Information Processing Systems}, pp.\
  4148--4158, 2017.

\bibitem[Yang et~al.(2019)Yang, Dai, Yang, Carbonell, Salakhutdinov, and
  Le]{xlnet}
Yang, Z., Dai, Z., Yang, Y., Carbonell, J.~G., Salakhutdinov, R., and Le, Q.~V.
\newblock Xlnet: Generalized autoregressive pretraining for language
  understanding.
\newblock \emph{CoRR}, abs/1906.08237, 2019.
\newblock URL \url{http://arxiv.org/abs/1906.08237}.

\bibitem[Yuxin~Wu(2018)]{yuxin2018groupnorm}
Yuxin~Wu, K.~H.
\newblock Group normalization.
\newblock \emph{arXiv preprint arXiv:1803.08494}, 2018.

\end{thebibliography}

\clearpage
\newpage
\appendix
\section{Supplementary material for~\Cref{sec:prelim}}
\label{app:prelim}
To better explain the hardware efficiency of the pipeline parallel training methods introduced in~\Cref{sec:methods}, we discuss the memory footprint and the throughput of the introduced methods in more details. Throughout the remainder of the appendix we use normalized throughput instead of pipeline utilization as the two are linearly proportional to each other. In~\Cref{sec:activation_mem}, we first discuss the activation memory which is the major component of memory consumption in pipeline-parallel training. We then propose a new gradient checkpointing method to trade moderate compute for significantly lower activation memory footprint in~\Cref{app:comp_act_tradeoff}, which is applicable to both the synchronous and asynchronous methods introduced in~\Cref{sec:methods}. Finally, we discuss the throughput of the synchronous (GPipe) and asynchronous (PipeDream and PipeMare) methods under the same budget for activation memory and compute (measured in FLOPs), which is used to estimate the time-to-accuracy across the paper. 

To discuss with consistent notations across methods, we define $M$ and $N$ respectively as the activation size per microbatch per neural net layer and the number of microbatches in each minibatch. We assume that we use models with $L$ layers, which are trained using a pipeline with $P$ stages. For clarity and simplicity in exposing the memory footprint and throughput, we assume that the model layers are partitioned equally across stages and the activation memory usage of each layer is the same.

\subsection{Activation Memory}
\label{sec:activation_mem}
\paragraph*{PipeMare and PipeDream}
PipeMare and PipeDream has the same amount of activation memory requirement. This is because in both scenarios, pipeline does not have bubbles or stalls; the activations are cached and utilized with the same pipeline behavior pattern. In particular, the activation memory cached by stage $i$ is proportional to the number of stages between forward and backward, i.e., $O(2(P-i)+1)$. Therefore, the total activation memory is

\begin{equation}
 \Act_{PM} = O(MPL).
 \label{eq:PipeMareMem}
\end{equation}

\paragraph*{GPipe} Here we discuss on the activation memory consumption of GPipe~\cite{huang2018gpipe}. When the activations of every layer in neural nets are cached for backpropagation, by multiplying the activation memory per minibatch per layer $B = MN$ with the number of layers L, we have the activation memory for GPipe as

 \begin{equation}
 	\Act_{GP} = O(MNL).
 	\label{eq: GpipeMem}
 \end{equation}
 
When re-materialization proposed by \cite{huang2018gpipe} is considered, we only need to store the activations of a minibatch at every stage boundary, and recompute the activations for all the layers inside the stage. Therefore the activation memory per stage is $O(MN + M\frac{L}{P})$, with the total activation memory reduced to:

\[
 \tilde{\Act}_{GP} = O(MNP + ML) = O(M(NP + L)).
\]

When $P \ll L$, the saving on activation memory is significant. However, in the fine-grain pipeline-parallel setting when $P \approx L$, the above equation goes back to \cref{eq: GpipeMem} and demonstrates negligible memory savings. This observation motivates us to propose the PipeMare recompute technique in~\Cref{app:comp_act_tradeoff}, which can apply to both synchronous (GPipe) and asynchronous (PipeMare) and effectively reduce the activation memory in fine-grained pipeline training.

\subsection{Trade compute for memory via PipeMare Recompute}
\label{app:comp_act_tradeoff}

In the fine-grain pipeline training setting, we have $P \approx L$. For simplicity in discussion, we assume $P=L$. In this setting, \cref{eq:PipeMareMem} becomes,

\begin{equation}
 \Act_{PM} = O(MP^2).
 \label{eq: PipeMareMemSqrt}
\end{equation}

\begin{figure}
	\centering
	\ifarxiv
	\includegraphics[width=0.4\linewidth]{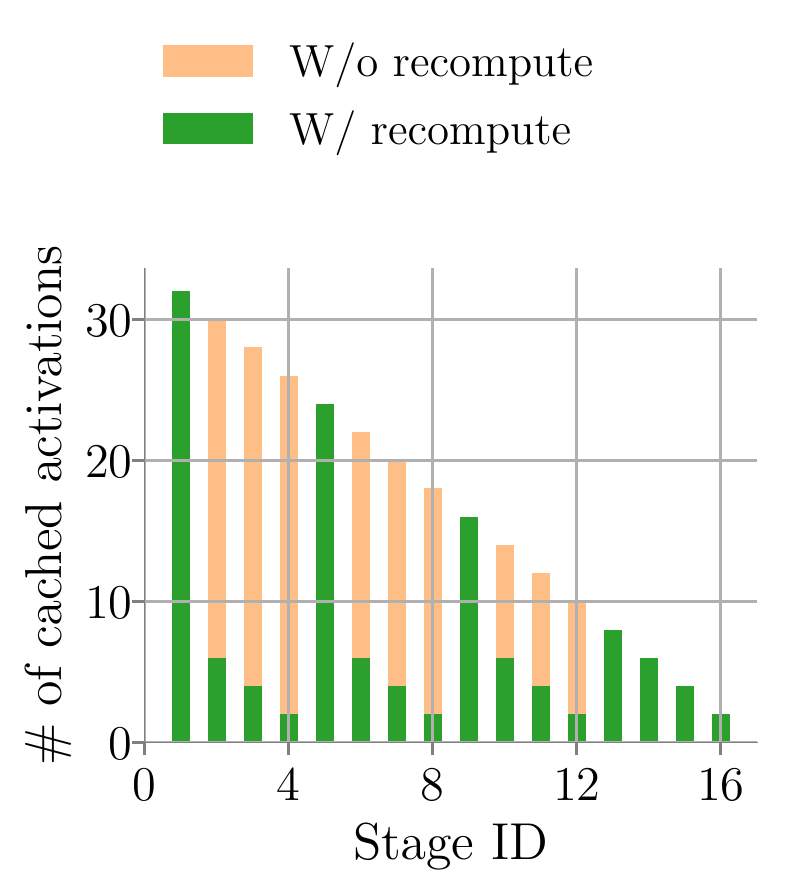}
	\else
	\includegraphics[width=0.9\linewidth]{figures/pipemare_recompute_memory.pdf}
	\fi
	\caption{Activation memory footprint of PipeMare recompute in each pipeline stage. In this plot, we demonstrate the \# of activations at each stage using an example with $16$ stages equally split into $4$ segments. The green bars in the plot stands for the memory consumption of each stage in terms of the number of microbatch activations copies in PipeMare with PipeMare recompute. The orange bars stands for the additional memory required when recompute is not used.}
	\label{fig:Recomputevisual}
\end{figure}

In other words, while throughput increases linearly with number of stages $P$, activation memory can scale quadratically. In order to reduce the memory pressure, here we propose a new way of utilizing recompute, to trade a small amount of compute resources for huge activation memory savings. Instead of recomputing the activations inside each stage \cite{huang2018gpipe}, we propose to recompute the activations across a segment of multiple stages, which we call PipeMare Recompute, to allow effective activation memory reduction in the fine-grain pipeline setting.

PipeMare Recompute utilizes a simple strategy. It recomputes the activation in advance so that the recomputed activation of the last stage in a segment arrives right at the time when the corresponding backpropagation needs to process this activation. Unlike the single-stage recompute proposed in GPipe~\cite{huang2018gpipe}, PipeMare Recompute does not stall the backpropagation operations as it can be overlapped with the forward and backward operations in the same pipeline stage.
To enable this overlap, we need to consume approximately $25\%$ of the total compute resources. Specifically, the pipeline needs to simultaneously compute for the forward, backward and recompute operations, with the backward operations consuming $2\times$ more compute than forward and recompute operations respectively. 

For the simplicity of demonstrating the activation memory saving attained by PipeMare Recompute, we assume $P=L$ in the fine-grain pipeline setting and group the stages into segments each with $S$ stages. Let us assume the $i-th$ stage is the beginning stage of a specific segment, then the memory consumption for this segment is $O(2(P-i) + S^2)$. As visualized in~\Cref{fig:Recomputevisual} for an example with $16$ stages and $4$ segments, the first term $2(P-i)$ in the segment-wise activation memory is for caching activations at the first stage in the segment for recompute. The second term $S^2$ then describes the memory buffers needed for recomputed activations that are used by backward pass (e.g., recompute of $j-th$ stage in a segment needs to start $2(S-j)$ steps earlier before the corresponding gradient arrives at this stage).
Consequently, given the memory consumption in each segment is $O(2(P-i) + S^2)$, the total memory with $P/S$ segments is determined by

 \[
  \Act_{PM}^r(S) = O\left(M(P + S^2) \cdot \frac{P}{S}\right) = O\left(MP(\frac{P}{S} + S)\right).
 \]

 When $S = \sqrt{P}$, we can get the minimum memory consumption,

 \begin{equation}
  \Act_{PM}^r = O\left(MP^{\frac{3}{2}}\right).
  \label{eq: PipeMareRecMem}
 \end{equation}

Note the quadratic dependency on $P$ in~\Cref{eq: PipeMareMemSqrt} is reduced to a power of $\frac{3}{2}$, indicating a significantly lowered activation memory in the fine-grain pipeline-parallelism with large $P$ values. 
 
We can similarly apply the PipeMare recompute technique to GPipe as well. In order to overlap recompute with forward and backward pass, each stage (except the first stage) in a segment needs to cache the same amount of activations as those of PipeMare. Whereas for the first stage in each segment, it needs to cache $N$ instead of $2(P-i)$ activations. This is because GPipe stalls at the boundary of minibatch, and there are $N$ microbatches to be processed in the minibatch. That being said, the activation memory of GPipe is $O(M(N + S^2)\cdot \frac{P}{S})$. Thus when $S = \sqrt{N}$, the minimum activation memory footprint of GPipe is

 \begin{equation}
  \Act_{GP}^r = O\left(MPN^{\frac{1}{2}}\right).
  \label{eq: GPipeRecMem}
 \end{equation}

We summarize the activation memory consumption with and without recompute for GPipe, PipeDream and PipeMare in~\Cref{tab: recompute_mem}. We can observe that for both synchronous and asynchronous pipeline-parallel training, the PipeMare Recompute can significantly reduce the activation memory in the fine-grain pipeline parallelism with large number of stages. The concrete activation memory saving of PipeMare on various tasks discussed in main text is shown in ~\cref{tab: activation memory saving}.

\begin{table*}[t]
	\centering
	\small
	\begin{tabular}{c c c}
	\toprule
	Mode & w/o PipeMare Recompute & w/ PipeMare Recompute \\
	\midrule
	GPipe & $MPN$ & $MPN^{\frac{1}{2}}$ \\
 	PipeMare/PipeDream & $MP^2$ & $MP^{\frac{3}{2}}$ \\
	\bottomrule	
	\end{tabular}
	\caption{Activation memory requirement by GPipe, PipeDream and PipeMare. Here we assume the total number of pipeline stages is the same as total neural network layers/operators, i.e., $P = L$. Note the activation memories for PipeMare and PipeDream are the same.}
	\label{tab: recompute_mem}	
\end{table*}

\begin{table*}[t]
  \centering
  \small
	\begin{tabular}{c c c c}
		\toprule
		Dataset & number of stages & Activation memory without recompute & Activation memory with recompute \\
		\midrule
		CIFAR10  & 107 & \multirow{4}{*}{1X} & \textbf{0.097X} \\
    ImageNet & 107 &                     &  \textbf{0.097X} \\
    IWSLT14  & 93  &                     &  \textbf{0.104X} \\
    WMT17    & 91  &                     &  \textbf{0.105X} \\
		\bottomrule
	\end{tabular}
  \caption{Activation memory of PipeMare for various tasks. Activation memory can be significantly reduced by using PipeMare Recompute.}
	\label{tab: activation memory saving}
\end{table*}

\section{Supplementary material for~\Cref{sec:theory}}
\label{app:theory}

\subsection{Motivating examples in deep learning}
\label{app:dl_motivate}

\begin{figure}[t]
\centering
\ifarxiv
\includegraphics[width=0.7\linewidth]{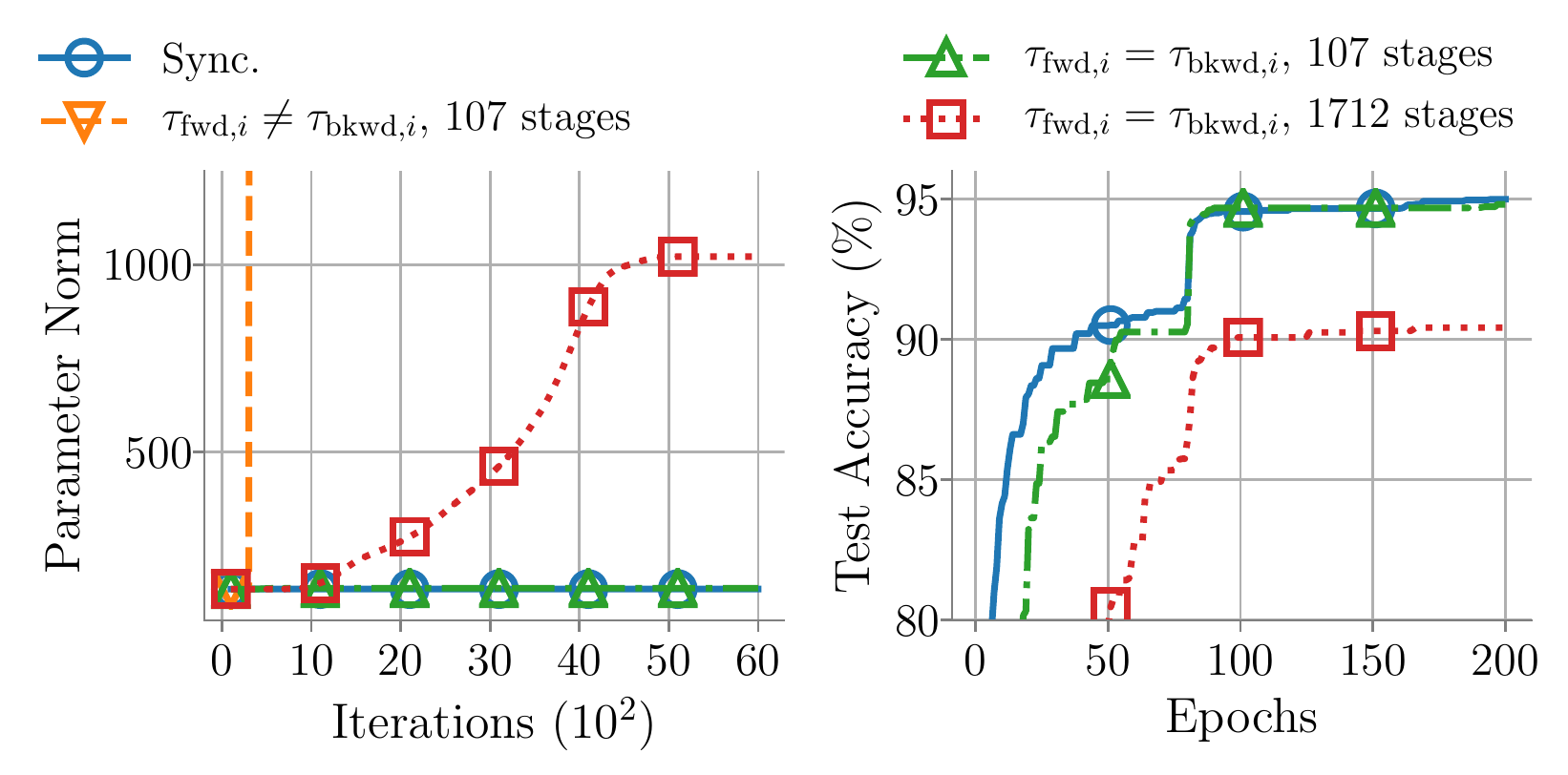}
\else
\includegraphics[width=0.99\linewidth]{figures/pipedream-test-accuracy-param-norm-cifar10.pdf}
\fi

\caption{Analysis on the divergence for asynchronous pipeline-parallel training: the divergence is caused by the forward delay $\tau_{\text{fwd},i}$; it is further exacerbated by forward-backward delay discrepancy when $\tau_{\text{fwd},i} \neq \tau_{\text{bkwd},i}$. Specifically in the left plot, we observe that using 1712 stages without forward-backward delay discrepancy, asynchronous training diverges at the beginning. 
We also observe that with 107 stages, asynchronous training diverges with forward-backward delay discrepancy, while it does not diverge without forward-backward delay discrepancy; this indicates that delay discrepancy can exacerbate the divergence behavior. These observations motivates us to explore the technique to stabilize asynchronous pipeline-parallel training.} 
\label{fig:theoryprelim}
\end{figure}

Figure~\ref{fig:theoryprelim}(a) illustrates that, just as we saw for the quadratic model, pipeline-parallel SGD can not be run naively with the same hyperparameters as would be used in the baseline model, since this would significantly negatively impact loss. 
Figure~\ref{fig:theoryprelim}(b) shows why: pipeline-parallel SGD is \emph{diverging} to infinity, completely failing to learn, even for a step size scheme for which the sequential model achieves state-of-the-art results.
This matches our results on the quadratic model.
For Resnet50 with standard hyperparameters, Figure~\ref{fig:theoryprelim} shows that this phenomenon is caused by the delay: the red series shows that, even when $\tau_{\text{fwd},i} = \tau_{\text{bkwd},i}$ in simulation, substantially large fixed delay can cause the system to diverge.
Figure~\ref{fig:theoryprelim} also illustrates that this divergence is exacerbated by forward-backward delay discrepancy: the orange series shows that even when the learning rate and delay $\tau_{\text{fwd},i}$ are kept the same, adding delay discrepancy can cause the algorithm to diverge.

\subsection{Proof of Lemma~\ref{lemmaRootsGD}}
We start by trying to find the $\alpha$ for which $p$ has a complex root on the unit circle.
Note that since $(1 - iy)/(1 + iy)$ always lies on the unit circle for any $y \in \R$, it suffices to find $\alpha$ for which
\[
	0
	=
	p\left( \frac{1 - iy}{1 + iy} \right)
	=
	\left( \left( \frac{1 - iy}{1 + iy} \right) - 1 \right) \left( \frac{1 - iy}{1 + iy} \right)^\tau + \alpha \lambda.
\]
for some $y > 0$.
After a little simplification, this becomes
\begin{equation}
	2iy \cdot \left( 1 - iy \right)^\tau
	=
	\alpha \lambda \cdot \left( 1 + iy \right)^{\tau + 1}.
	\label{eqnPolyGD1}
\end{equation}
Next, we take the argument. Since $y$, $\alpha$, and $\lambda$ are real and positive, for some $n \in \Z$,
\[
	\frac{\pi}{2} + 2 \pi n + \tau \Arg\left( 1 - iy \right) = (\tau + 1) \Arg\left( 1 + iy \right),
\]
which implies that, since $\Arg\left( 1 - iy \right) = -\Arg\left( 1 + iy \right)$,
\[
	\Arg\left( 1 + iy \right) = \frac{\pi + 4 \pi n}{4 \tau + 2}.
\]
This uniquely determines the value of $y$, because $y = \tan \Arg(1 + iy)$.
To get the corresponding value of $\alpha$, notice that if we take the magnitude of (\ref{eqnPolyGD1}), it simplifies to
\[
	\alpha \lambda = \frac{|2iy|}{|1 + i y|} = \frac{2y}{\sqrt{1 + y^2}} = 2 \sin \Arg(1 + i y),
\]
so there can be a point on the unit circle when
\[
	\alpha = \frac{2}{\lambda} \cdot \sin\left( \frac{\pi + 4 \pi n}{4 \tau + 2} \right)
\]
for any $n \in Z$.
The lemma statement now follows directly from a root-counting argument.

The main components of the root-counting argument are as follows.
First, notice that for small $\alpha$, all the roots of $p$ will be within the interior of the unit disk, since as $\alpha$ approaches $0$ from above, all but one of the roots will approach $0$ and the remaining root will approach $1$ from the left.
To see this, notice that when $\alpha = 0$,
\[
	p(\omega) = (\omega - 1) \cdot \omega^{\tau}.
\]
On the other hand, as $\alpha \rightarrow \infty$, all the roots will diverge in magnitude to infinity, which means they must eventually leave the unit circle.
To see why, notice that any root of $p$ must satisfy
\[
	0 = p(\omega) = \omega^{\tau + 1} - \omega^{\tau} + \alpha \lambda,
\]
which implies from taking the magnitude that
\[
	| \omega |^{\tau + 1} + | \omega |^{\tau} \ge \alpha \lambda.
\]
Thus, we can conclude that all $\tau + 1$ roots of the polynomial $p$ must pass through the unit circle as $\alpha$ moves from $0_+$ to $\infty$.

Now, from the proof of Lemma~\ref{lemmaRootsGD}, we know exactly where these crossings of the unit circle can occur.
They happen for
\[
	\alpha = \frac{2}{\lambda} \cdot \sin\left( \frac{\pi + 4 \pi n}{4 \tau + 2} \right),
\]
and at a point $\omega$ on the unit circle with
\[
	\Arg(\omega) = \pm \frac{\pi + 4 \pi n}{4 \tau + 2}.
\]
Not all values of $n$ correspond to a positive value of $\alpha$, and many values of $n$ will result in the same value of $\alpha$.
Clearly we can restrict our attention to $0 \le n < 2 \tau + 1$, since adding $2 \tau + 1$ to $n$ results in the same values for $\alpha$ and $\omega$.
The step size $\alpha$ will only be positive when, for some $m \in \Z$,
\[
	\frac{\pi + 4 \pi n}{4 \tau + 2} + 2 \pi m \in \left(0, \pi \right),
\]
since this is where the $\sin$ is positive.
Dividing both sides by $\pi$ and multiplying by $2 \tau + 1$, this happens when
\[
	\frac{1}{2} + 2 n + 2 m (2 \tau + 1) \in \left(0, 2 \tau + 1 \right).
\]
In other words, this will happen for $n \in \{0, 1, \ldots, \tau\}$.
However, half of these produce redundant values of $\alpha$, since
\begin{align*}
	\sin\left( \frac{\pi + 4 \pi (\tau - n)}{4 \tau + 2} \right)
	&=
	\sin\left( \frac{\pi (4 \tau + 2) - \pi - 4 \pi n}{4 \tau + 2} \right) \\
	&=
	\sin\left( \pi - \frac{\pi + 4 \pi n}{4 \tau + 2} \right) \\
	&=
	\sin\left( \frac{\pi + 4 \pi n}{4 \tau + 2} \right).
\end{align*}
So we can restrict our attention to $0 \le n \le \frac{\tau}{2}$.
If $\tau$ is odd, then each of these assignments of $n$ corresponds to two roots on the unit circle.
If $\tau$ is even, then each of these assignments corresponds to two roots, except for the assignment $n = \frac{\tau}{2}$, for which
\[
	\Arg(\omega) = \pm \frac{\pi + 2 \pi \tau}{4 \tau + 2} = \frac{\pi}{2}
\]
corresponds to only one root on the unit circle.
Thus there are only ever $\tau + 1$ assignments of $(\alpha, \omega)$ for which $\omega$ is a root on the unit circle of
\[
	0 = (\omega - 1) \cdot \omega^{\tau} + \alpha \lambda.
\]
Furthermore, none of those roots can be multiple roots, because if they were multiple roots they would need to be zeros of the polynomial $p'(\omega)$, and none of the roots of that polynomial lie on the unit disk.
As a result, every root crossing of the unit disk must involve only a single root.
Since there are $\tau + 1$ roots and $\tau + 1$ opportunities for a crossing, and all $\tau + 1$ roots \emph{must} cross at some point, each crossing of the unit circle must correspond to a root moving \emph{out} of the unit disk.
As a consequence, no root can ever move back in to the unit disk, since there is no room for it to do so.
Thus, after the first roots leave the unit disk at
\[
	\alpha = \frac{2}{\lambda} \cdot \sin\left( \frac{\pi}{4 \tau + 2} \right),
\]
there is never a time at which all the roots are inside the unit disk.

Finally, recall that $p$ can have a double root only where its first derivative $p'$ has a root.
This will occur only where
\[
	p'(w) = (\tau + 1) \omega^{\tau} - \tau \omega^{\tau-1} = 0,
\]
which happens at
\[
	\omega = \frac{\tau}{\tau + 1}.
\]
This corresponds to a value of $\alpha$ of
\begin{align*}
	\alpha 
	&=
	\frac{1}{\lambda} (1 - \omega) \omega^{\tau} \\
	&=
	\frac{1}{\lambda (\tau + 1)} \left(\frac{\tau}{\tau + 1} \right)^{\tau}.
\end{align*}

This proves the lemma.

\subsection{An extension to SGD with momentum.}
\label{app:ext_mom}
Deep neural networks are often trained with momentum~\cite{sutskever2013importance}.
A natural question is whether the $O(\tau^{-1})$ stability threshold also holds if momentum is used.
When we add momentum, our update step becomes
\begin{align*}
    v_{t+1} = \beta v_t - \alpha \nabla f_t(u_{\text{fwd},t}, u_{\text{bkwd},t}), \;
    w_{t+1} = w_t + v_{t+1}.
\end{align*}
We make the same simplifying assumptions as we made above for the non-momentum case, assuming a constant $\tau$ and quadratic loss.
This results in an update step that, just as above, can be expressed in terms of a companion matrix which will have characteristic polynomial
\begin{equation}
    p(\omega) = \omega^{\tau + 1} - (1 + \beta) \omega^{\tau} + \beta \omega^{\tau - 1} + \alpha \lambda.
    \label{eqnCharPolyMom}
\end{equation}
As in the non-momentum case, we can analyze this for stability by finding the parameters for which the roots of $p$ lie inside the unit disk.
\begin{lemma}
\label{lemmaRootsMom}
For any momentum parameter $0 < \beta \le 1$, there exists a step size $\alpha$ with
\[
    0 < \alpha \le \frac{4}{\lambda} \cdot \sin\left( \frac{\pi}{4\tau + 2} \right)
\]
such that at least one of the roots of the polynomial $p$ of (\ref{eqnCharPolyMom}) lies outside the interior of the unit disk.
\end{lemma}
This lemma shows that adding momentum does not let us escape from the $O(\tau^{-1})$ step size requirement observed for SGD.
It suggests that the $O(\tau^{-1})$ threshold is general and not just specific to plain SGD, and it motivates our use of Technique 1 with \emph{all learning algorithms}, not just SGD.

We make the same simplifying assumptions as we made above for the non-momentum case, assuming a constant $\tau$ and quadratic loss.
This results in an update step of
\[
	w_{t+1} - w_t = \beta \left( w_t - w_{t-1} \right) - \alpha \lambda w_{t-\tau} + \alpha \eta_t.
\]
Just as in the non-momentum case, we can write this in terms of a companion matrix, which will have characteristic polynomial
\begin{equation}
	p(\omega) = \omega^{\tau + 1} - (1 + \beta) \omega^{\tau} + \beta \omega^{\tau - 1} + \alpha \lambda.
	\label{eqnCharPolyMom}
\end{equation}
As in the non-momentum case, we will analyze this for stability by finding the parameters for which the roots of $p$ lie inside the unit disk.

To prove the lemma, we start with the expression for the polynomial
\begin{align*}
	p(\omega) 
	&=
	\omega^{\tau + 1} - (1 + \beta) \omega^{\tau} + \beta \omega^{\tau - 1} + \alpha \lambda \\
	&=
	(\omega - \beta) \cdot (\omega - 1) \cdot \omega^{\tau - 1} + \alpha \lambda.
\end{align*}
As for the non-momentum case, we consider the substitution
\[
	\omega = \frac{1 - iy}{1 + iy},
\]
which always lies on the unit circle for any $y \in \R$.
(Without loss of generality, we consider $y > 0$, which corresponds to roots in the lower half-plane. This is without loss of generality because, since $p$ is a real polynomial, its complex roots always appear in pairs.)
We want to find $\alpha$ and $\beta$ for which
\begin{align*}
	0
	&=
	p\left( \frac{1 - iy}{1 + iy} \right) \\
	&=
	\left( \left( \frac{1 - iy}{1 + iy} \right) - \beta \right) 
	\left( \left( \frac{1 - iy}{1 + iy} \right) - 1 \right)
	\left( \frac{1 - iy}{1 + iy} \right)^{\tau-1}
	\\&\hspace{2em}+
	\alpha \lambda.
\end{align*}
This can be simplified to
\begin{align*}
	0
	&=
	\left( 1 - \beta \cdot \frac{1 + iy}{1 - iy} \right) 
	\left( \frac{-2iy}{1 + iy} \right)
	\left( \frac{1 - iy}{1 + iy} \right)^{\tau}
	+
	\alpha \lambda,
\end{align*}
and so
\[
	\left( 1 - \beta \cdot \frac{1 + iy}{1 - iy} \right)
	\cdot
	2iy
	\cdot
	( 1 - iy )^{\tau}
	=
	\alpha \lambda (1 + i y)^{\tau+1}.
\]
Define $\theta$ as
\[
	\theta = \Arg\left( 1 - \beta \cdot \frac{1 + iy}{1 - iy} \right) + \frac{\pi}{2}.
\]
Notice that since the thing inside the $\Arg$ is $1$ minus something with magnitude less than $1$ times something that is on the unit circle in the upper half plane, it will necessarily end up in the fourth quadrant, and so
\[
	\theta - \frac{\pi}{2} \in \left(-\frac{\pi}{2}, 0 \right)
	\Rightarrow
	\theta \in \left(0, \frac{\pi}{2} \right).
\]
Now taking the argument of the whole expression gives us, for any $n \in \Z$,
\[
	\theta
	+
	2 \pi n
	+
	\tau \Arg( 1 - iy )
	=
	(\tau + 1) \Arg(1 + i y),
\]
which simplifies to
\[
	\Arg(1 + i y) = \frac{\theta + 2 \pi n}{2 \tau + 1}.
\]
In this case,
\[
	y = \tan\left( \frac{\theta + 2 \pi n}{2 \tau + 1} \right).
\]
Next, we derive an expression for $\beta$. Since
\begin{align*}
	\theta 
	&=
	\Arg\left( 1 - \beta \cdot \frac{1 + iy}{1 - iy} \right) + \frac{\pi}{2} \\
	&=
	\Arg\left( \frac{(1 - \beta) - iy (1 + \beta)}{1 - iy} \right) + \frac{\pi}{2} \\
	&=
	\Arg\left( (1 - \beta) - iy (1 + \beta) \right) + \Arg(1 + i y) + \frac{\pi}{2} \\
	&=
	\Arg\left( 1 + i \frac{1 - \beta}{y (1 + \beta)} \right) + \Arg(1 + i y),
\end{align*}
so
\[
	\theta - \frac{\theta + 2 \pi n}{2 \tau + 1}
	=
	\Arg\left( 1 + i \frac{1 - \beta}{y (1 + \beta)} \right),
\]
and
\begin{align*}
	\frac{1 - \beta}{1 + \beta} 
	&=
	\tan\left( \frac{\theta + 2 \pi n}{2 \tau + 1} \right)
	\tan\left( \theta - \frac{\theta + 2 \pi n}{2 \tau + 1} \right) \\
	&=
	\frac{
		\cos\left( \theta - \frac{2 \theta + 4 \pi n}{2 \tau + 1} \right)
		-
		\cos(\theta)
	}{
		\cos\left( \theta - \frac{2 \theta + 4 \pi n}{2 \tau + 1} \right)
		+
		\cos(\theta)
	}.
\end{align*}
Now taking the absolute value to find $\alpha$ gives us
\begin{align*}
	\alpha \lambda
	&=
	\left| 1 - \beta \cdot \frac{1 + iy}{1 - iy} \right|
	\cdot
	\frac{2y}{|1 + i y|} \\
	&=
	2 \cdot \left| 1 - \beta \cdot \frac{1 + iy}{1 - iy} \right|
	\cdot
	\sin\left( \frac{\theta + 2 \pi n}{2 \tau + 1} \right).
\end{align*}
Next, consider the case where $n = 0$.
In this case,
\begin{align*}
	\frac{1 - \beta}{1 + \beta} 
	&=
	\frac{
		\cos\left( \theta - \frac{2 \theta}{2 \tau + 1} \right)
		-
		\cos(\theta)
	}{
		\cos\left( \theta - \frac{2 \theta}{2 \tau + 1} \right)
		+
		\cos(\theta)
	}.
\end{align*}
It is clear that there is a one-to-one relationship between accessible $\theta$ and $\beta$ here, because we can represent $\beta = 0$ with $\theta = \pi/2$, and $\beta = 1$ with $\theta = 0$.
So, for every $\beta$ (and given a fixed $\tau$), we can find a $\theta$ that satisfies this equation.
Using that $\theta$, we can then assign
\[
	y = \tan\left( \frac{\theta}{2 \tau + 1} \right).
\]
Since $\theta$ is bounded, $y$ is guaranteed to be in range.
So, the equation
\[
	0 = p\left( \frac{1 - iy}{1 + iy} \right)
\]
will be guaranteed to hold for some $\alpha$.
This $\alpha$ will be given by
\begin{align*}
	\alpha \lambda
	&=
	2 \cdot \left| 1 - \beta \cdot \frac{1 + iy}{1 - iy} \right|
	\cdot
	\sin\left( \frac{\theta}{2 \tau + 1} \right).
\end{align*}
So, since $\beta < 1$, it follows that this $\alpha$ will satisfy
\[
	\alpha
	\le
	\frac{4}{\lambda} \cdot \sin\left( \frac{\theta}{2 \tau + 1} \right)
	\le
	\frac{4}{\lambda} \cdot \sin\left( \frac{\pi}{4 \tau + 2} \right),
\]
which is what we wanted to show.
This proves that for any $\beta$, there exists a $\alpha$ at least this large for which the algorithm is unstable.

\subsection{Proof of Lemma~\ref{lemmaRootsDD}}
We know, from our baseline analysis, that when
\[
    \alpha = \frac{2}{\lambda} \cdot \sin\left( \frac{\pi}{4\tau + 2} \right)
\]
and $\Delta = 0$, the polynomial $p$ has a root at
\[
    \omega = \exp\left( \frac{i \pi}{2 \taufwd + 1} \right).
\]
Consider values of $\alpha$ and $\Delta$ for which $p$ would have a root at
\[
    \omega = \exp( i \theta )
\]
for
\[
    \theta \in \left( 0, \frac{\pi}{2 \taufwd + 1} \right].
\]
In this case, we'd have
\begin{align*}
    0
    &=
    \exp(i \taufwd \theta) \cdot (\omega - 1)
    \\&\hspace{2em}-
    \alpha \cdot \Delta \cdot \exp(i (\taufwd - \taubkwd) \theta)
    \\&\hspace{2em}+
    \alpha \cdot (\lambda + \Delta),
\end{align*}
which is equivalent to
\begin{align*}
    0
    &=
    \exp\left(i \frac{\taufwd + \taubkwd}{2} \theta \right) \cdot (\omega - 1)
    \\&\hspace{2em}-
    \alpha \cdot \Delta \cdot \exp\left(i \frac{\taufwd - \taubkwd}{2} \theta \right)
    \\&\hspace{2em}+
    \alpha \cdot (\lambda + \Delta) \cdot \exp\left(-i \frac{\taufwd - \taubkwd}{2} \theta \right).
\end{align*}
If we take the real part of this, we get
\begin{align*}
    0
    &=
    \cos\left(\frac{\taufwd + \taubkwd + 2}{2} \cdot \theta \right)
    \\&\hspace{2em}-
    \cos\left(\frac{\taufwd + \taubkwd}{2} \cdot \theta \right)
    \\&\hspace{2em}+
    \alpha \lambda \cos\left(\frac{\taufwd - \taubkwd}{2} \theta \right) \\
    &=
    -2
    \sin\left(
        \frac{\tau_{\text{fwd}} + \tau_{\text{bkwd}} + 1}{2} \cdot \theta 
    \right)
    \cdot
    \sin\left(
        \frac{\theta}{2}
    \right)
    \\&\hspace{2em}+
    \alpha \lambda \cos\left(\frac{\taufwd - \taubkwd}{2} \theta \right),
\end{align*}
so solving for $\alpha$ gives us
\[
    \alpha
    =
    \frac{
        2
        \sin\left(
            \frac{\tau_{\text{fwd}} + \tau_{\text{bkwd}} + 1}{2} \cdot \theta 
        \right)
        \cdot
        \sin\left(
            \frac{\theta}{2}
        \right)
    }{
        \lambda \cos\left(\frac{\taufwd - \taubkwd}{2} \cdot \theta \right)
    }.
\]
On the other hand, if we take the imaginary part instead of the real part, we get
\begin{align*}
    0
    &=
    \sin\left(\frac{\taufwd + \taubkwd + 2}{2} \cdot \theta \right)
    \\&\hspace{2em}-
    \sin\left(\frac{\taufwd + \taubkwd}{2} \cdot \theta \right)
    \\&\hspace{2em}-
    \alpha (\lambda + 2 \Delta) \sin\left(\frac{\taufwd - \taubkwd}{2} \theta \right) \\
    &=
    2
    \cos\left(
        \frac{\tau_{\text{fwd}} + \tau_{\text{bkwd}} + 1}{2} \cdot \theta 
    \right)
    \cdot
    \sin\left(
        \frac{\theta}{2}
    \right)
    \\&\hspace{2em}-
    \alpha (\lambda + 2 \Delta) \sin\left(\frac{\taufwd - \taubkwd}{2} \theta \right)  \\
    &=
    2
    \cos\left(
        \frac{\tau_{\text{fwd}} + \tau_{\text{bkwd}} + 1}{2} \cdot \theta 
    \right)
    \cdot
    \sin\left(
        \frac{\theta}{2}
    \right)
    \\&\hspace{2em}-
    (\lambda + 2 \Delta) \sin\left(\frac{\taufwd - \taubkwd}{2} \theta \right)
    \\&\hspace{4em}\cdot
    \frac{
        2
        \sin\left(
            \frac{\tau_{\text{fwd}} + \tau_{\text{bkwd}} + 1}{2} \cdot \theta 
        \right)
        \cdot
        \sin\left(
            \frac{\theta}{2}
        \right)
    }{
        \lambda \cos\left(\frac{\taufwd - \taubkwd}{2} \theta \right)
    } \\
    &=
    1
    -
    \left(1 + \frac{2 \Delta}{\lambda}\right) \tan\left(\frac{\taufwd - \taubkwd}{2} \theta \right)
    \\&\hspace{2em}\cdot
    \tan\left(
        \frac{\tau_{\text{fwd}} + \tau_{\text{bkwd}} + 1}{2} \cdot \theta 
    \right).
\end{align*}
and so
\begin{align*}
    \frac{2 \Delta}{\lambda}
    &=
    \cot\left(\frac{\taufwd - \taubkwd}{2} \cdot \theta \right)
    \\&\hspace{2em}\cdot
    \cot\left(
        \frac{\tau_{\text{fwd}} + \tau_{\text{bkwd}} + 1}{2} \cdot \theta 
    \right)
    -
    1 \\
    &=
    \csc\left(\frac{\taufwd - \taubkwd}{2} \cdot \theta \right)
    \\&\hspace{2em}\cdot
    \csc\left(
        \frac{\tau_{\text{fwd}} + \tau_{\text{bkwd}} + 1}{2} \cdot \theta 
    \right)
    \\&\hspace{2em}\cdot
    \cos\left(
        \frac{2 \tau_{\text{fwd}} + 1}{2} \cdot \theta 
    \right).
\end{align*}
One thing we notice immediately from this expression is that it approaches infinity as $\theta \rightarrow 0^+$, goes to zero at
\[
    \theta = \frac{\pi}{2 \taufwd + 1},
\]
and is continuous and positive in between.
This means that all non-negative values of $\Delta$ are actually attained for some $\theta$, and there is a one-to-one mapping between $\Delta$ and $\theta$ in this interval.
Furthermore, since $\alpha$ approaches $0$ monotonically as $\theta$ approaches $0$ over this interval, this means that there is no absolute lower bound on how small $\alpha$ can get.
So all we need is a bound on $\alpha$ in terms of $\Delta$.

In the limit of small $\theta$,
\[
    \frac{2 \Delta}{\lambda}
    =
    \left(\frac{\taufwd - \taubkwd}{2} \cdot \theta \right)^{-1}
    \cdot
    \left(
        \frac{\tau_{\text{fwd}} + \tau_{\text{bkwd}} + 1}{2} \cdot \theta 
    \right)^{-1}
\]
and
\begin{align*}
    \alpha
    &=
    \frac{
        2
        \left(
            \frac{\tau_{\text{fwd}} + \tau_{\text{bkwd}} + 1}{2} \cdot \theta 
        \right)
        \cdot
        \left(
            \frac{\theta}{2}
        \right)
    }{
        \lambda
    } \\
    &=
    \frac{
        2
        \left(
            \frac{\tau_{\text{fwd}} + \tau_{\text{bkwd}} + 1}{2} \cdot \theta 
        \right)
        \cdot
        \left(
            \frac{\theta}{2}
        \right)
    }{
        \lambda
    }
    \cdot
    \frac{\lambda}{2 \Delta}
    \\&\hspace{2em}\cdot
    \left(\frac{\taufwd - \taubkwd}{2} \cdot \theta \right)^{-1}
    \cdot
    \left(
        \frac{\tau_{\text{fwd}} + \tau_{\text{bkwd}} + 1}{2} \cdot \theta 
    \right)^{-1} \\
    &=
    \frac{1}{\Delta \cdot \left(\taufwd - \taubkwd \right)}.
\end{align*}
Can we get a real bound that matches this?
\begin{align*}
    \frac{\lambda}{\Delta}
    &=
    2
    \sin\left(\frac{\taufwd - \taubkwd}{2} \cdot \theta \right)
    \\&\hspace{2em}\cdot
    \sin\left(
        \frac{\tau_{\text{fwd}} + \tau_{\text{bkwd}} + 1}{2} \cdot \theta 
    \right)
    \\&\hspace{2em}\cdot
    \sec\left(
        \frac{2 \tau_{\text{fwd}} + 1}{2} \cdot \theta 
    \right) \\
    &=
    \left(
        \cos\left(
            \frac{2 \tau_{\text{bkwd}} + 1}{2} \cdot \theta 
        \right)
        -
        \cos\left(
            \frac{2 \tau_{\text{fwd}} + 1}{2} \cdot \theta 
        \right)
    \right)
    \\&\hspace{2em}\cdot
    \sec\left(
        \frac{2 \tau_{\text{fwd}} + 1}{2} \cdot \theta 
    \right) \\
    &=
    \frac{
        \cos\left(
            \frac{2 \tau_{\text{bkwd}} + 1}{2} \cdot \theta 
        \right)
    }{
        \cos\left(
            \frac{2 \tau_{\text{fwd}} + 1}{2} \cdot \theta 
        \right)
    }
    -
    1,
\end{align*}
so
\[
    1 + \frac{\lambda}{\Delta}
    =
    \frac{
        \cos\left(
            \frac{2 \tau_{\text{bkwd}} + 1}{2} \cdot \theta 
        \right)
    }{
        \cos\left(
            \frac{2 \tau_{\text{fwd}} + 1}{2} \cdot \theta 
        \right)
    }.
\]
It can be shown that for any $y < x < \frac{\pi}{2}$,
\[
    \frac{\cos(y)}{\cos(x)} \ge 1 + \frac{x^2 - y^2}{2}.
\]
(To see why, observe that for any $a \in [0,1]$, the third derivative of $\cos(ax) \cdot \sec(x)$ is non-negative over $x \in [0, \pi/2]$.)
So,
\[
    \frac{\lambda}{\Delta}
    \ge
    \frac{1}{2}
    \cdot
    \left(
        \left(
            \frac{2 \tau_{\text{fwd}} + 1}{2}
        \right)^2
        -
        \left(
            \frac{2 \tau_{\text{bkwd}} + 1}{2}
        \right)^2
    \right)
    \cdot
    \theta^2.
\]
Similarly, we have
\begin{align*}
    \alpha
    &=
    \frac{2}{\lambda} \cdot \sin\left( \frac{\theta}{2} \right)
    \\&\hspace{2em} \cdot
    {\scriptstyle \frac{
        (\lambda + \Delta) \cdot \sin\left(
            \frac{2 \tau_{\text{fwd}} + 1}{2} \cdot \theta 
        \right)
        +
        \Delta
        \cdot
        \sin\left(
            \frac{2 \tau_{\text{bkwd}} + 1}{2} \cdot \theta 
        \right)
    }{
        \lambda + 2 \Delta
    }} \\
    &\le
    \frac{1}{\lambda}
    \cdot
    \frac{
        (\lambda + \Delta) \cdot \left(
            \frac{2 \tau_{\text{fwd}} + 1}{2}
        \right)
        +
        \Delta
        \cdot
        \left(
            \frac{2 \tau_{\text{bkwd}} + 1}{2}
        \right)
    }{
        \lambda + 2 \Delta
    }
    \cdot
    \theta^2 \\
    &\le
    \frac{2}{\Delta}
    \cdot
    \frac{
        (\lambda + \Delta) \cdot \left(
            \frac{2 \tau_{\text{fwd}} + 1}{2}
        \right)
        +
        \Delta
        \cdot
        \left(
            \frac{2 \tau_{\text{bkwd}} + 1}{2}
        \right)
    }{
        \lambda + 2 \Delta
    }
    \\&\hspace{2em}\cdot
    \left(
        \left(
            \frac{2 \tau_{\text{fwd}} + 1}{2}
        \right)^2
        -
        \left(
            \frac{2 \tau_{\text{bkwd}} + 1}{2}
        \right)^2
    \right)^{-1} \\
    &\le
    \frac{2}{\Delta}
    \cdot
    \left(
        \left(
            \frac{2 \tau_{\text{fwd}} + 1}{2}
        \right)
        +
        \left(
            \frac{2 \tau_{\text{bkwd}} + 1}{2}
        \right)
    \right)
    \\&\hspace{2em}\cdot
    \left(
        \left(
            \frac{2 \tau_{\text{fwd}} + 1}{2}
        \right)^2
        -
        \left(
            \frac{2 \tau_{\text{bkwd}} + 1}{2}
        \right)^2
    \right)^{-1} \\
    &\le
    \frac{2}{\Delta}
    \cdot
    \left(
        \left(
            \frac{2 \tau_{\text{fwd}} + 1}{2}
        \right)
        -
        \left(
            \frac{2 \tau_{\text{bkwd}} + 1}{2}
        \right)
    \right)^{-1} \\
    &\le
    \frac{2}{\Delta \cdot \left(\taufwd - \taubkwd \right)}.
\end{align*}
And this is an actual guarantee.
So, we've proven that for any $\Delta \ge 0$, there exists an $\alpha$ with
\[
    0
    <
    \alpha
    \le
    \frac{2}{\Delta \cdot \left(\taufwd - \taubkwd \right)}
\]
such that the polynomial $p$ has a root on the unit circle.
The other part of the $\min$ in the lemma statement follows directly from our original bound and the monotonicity of $\Delta$ and $\alpha$ in terms of $\theta$ over the interval we have been looking at.

\subsection{Justification for Claims in \Cref{sec:delay_discrepancy}}

In \Cref{sec:delay_discrepancy}, we motivated our choice of $\Delta$ by claiming that the second-order Taylor expansion of the characteristic polynomial of the companion matrix associated with momentum-corrected asynchronous pipeline-parallel SGD on the quadratic model around $\omega = 1$ is invariant to the delay-discrepancy-sensitivity parameter $\Delta$ if $\gamma$ is set appropriately.
Here, we justify that assertion, as well as the other assertions we made in that subsection.
First, we want to show formally that $\omega = 1$ is the ``interesting'' region.
We do this with the following lemma.

\begin{lemma}
For any polynomial functions $f$, $g$, and $h$, and any integer $\tau$, define the polynomial
\[
	p_{\tau}(\omega) = (\omega - 1) \cdot f(\omega) \cdot \omega^{\tau} - \alpha \cdot g(\omega) \cdot \omega^{\tau} - \alpha \cdot h(\omega), 
\]
and suppose that $f$ does not vanish anywhere on the unit circle.
For any $\tau$, let $\alpha_{\textnormal{thresh}}(\tau)$ be the smallest $\alpha > 0$ for which $p_{\tau}$ has a root on the unit circle, and let $\omega_{\textnormal{thresh}}(\tau)$ be one of the corresponding roots.
Then, if
\[
	\lim_{\tau \rightarrow \infty} \alpha_{\textnormal{thresh}}(\tau) = 0,
\]
then
\[
	\lim_{\tau \rightarrow \infty} \omega_{\textnormal{thresh}}(\tau) = 1.
\]
\label{lemmaOmega1}
\end{lemma}
\begin{proof}
Suppose that $p_{\tau}(\omega) = 0$ for some $\omega$ on the unit circle.
Solving for $\alpha$ gives
\begin{align*}
	\alpha 
	&=
	(\omega - 1) \cdot \frac{ f(\omega) \cdot \omega^{\tau} }{ g(\omega) \cdot \omega^{\tau} - h(\omega) } \\
	&=
	\Abs{\omega - 1} \cdot \frac{ \Abs{ f(\omega) } }{ \Abs{ g(\omega) \cdot \omega^{\tau} - h(\omega) } } \\
	&\ge
	\Abs{\omega - 1} \cdot \frac{ \Abs{ f(\omega) } }{ \Abs{ g(\omega) } + \Abs{ h(\omega) } } \\
	&\ge
	\Abs{\omega - 1} \cdot \frac{ f_{\min} }{ g_{\max} + h_{\max} },
\end{align*}
where these $\min$ and $\max$ are taken over the unit circle.
So, for some constant $C > 0$ independent of $\tau$,
\[
	\Abs{\omega - 1} \le C \cdot \alpha
\]
(we know such a $C$ exists because $f$ does not vanish on the unit circle).
The lemma statement follows directly.
\end{proof}

This lemma shows in a very general sense that the points at which the roots of the characteristic polynomial first cross the unit circle as $\alpha$ increases from $0$ will approach $\omega = 1$ as $\tau$ approaches $\infty$.
Since we know from observation that for the systems we are studying, the smallest $\alpha$ at which the polynomial becomes unstable becomes smaller as $\tau$ approaches $\infty$, it follows that as $\tau \rightarrow \infty$, the points $\omega$ at which the system first becomes unstable must also approach $\omega = 1$.
This formally justifies our notion of the area where the ``action happens'' for large $\tau$.

Now, we will prove that the characteristic polynomial of the companion matrix associated with momentum-corrected asynchronous pipeline-parallel SGD on the quadratic model around $\omega = 1$ is invariant to the delay-discrepancy-sensitivity parameter $\Delta$ if $\gamma$ is set such that
\[
	\gamma = 1 - \frac{2}{\taufwd - \taubkwd + 1}.
\] 
Here, we justify that assertion.
First, observe that the characteristic polynomial of the companion matrix is
\begin{align*}
	p(\omega)
	&=
	(\omega - 1) (\omega - \gamma) \omega^{\taufwd}
	\\&\hspace{2em}+
	\alpha (\lambda + \Delta) (\omega - \gamma)
	\\&\hspace{2em}-
	\alpha \Delta \omega^{\taufwd - \taubkwd} (\omega - \gamma)
	\\&\hspace{2em}+
	\alpha \Delta \omega^{\taufwd - \taubkwd} (\taufwd - \taubkwd) (1 - \gamma) (\omega - 1).
\end{align*}
This can be seen by constructing the companion matrix from the update rule directly.

Notice that this polynomial satisfies all the conditions of the statement of Lemma~\ref{lemmaOmega1}, for appropriate values of $f$, $g$, and $h$, and letting $\tau = \taufwd - \taubkwd$.
At $\omega = 1$, we have
\[
	p(1) = \alpha \lambda (1 - \gamma)
\]
and
\[
	p'(1) = \alpha \lambda + 1 - \gamma,
\]
both of which are independent of the sensitivity parameter $\Delta$.
On the other hand, the second derivative is
\begin{align*}
	p''(1)
	&=
	2 \taufwd (1 - \gamma) + 2
	\\&\hspace{-2em}-
	\alpha \Delta (\taufwd - \taubkwd) (1 + \gamma - (1 - \gamma) (\taufwd - \taubkwd)).
\end{align*}
From here, notice that the $\Delta$ term drops out of this expression if we set $\gamma$ such that
\[
	0 = 1 + \gamma - (1 - \gamma) (\taufwd - \taubkwd);
\]
this occurs when
\begin{equation}
	\gamma = 1 - \frac{2}{\taufwd - \taubkwd + 1}.
	\label{eqnGammaValue}
\end{equation}
Also notice that in the limit of large $\tau$, we would have
\begin{align*}
	\operatorname{D}
	&=
	\gamma^{\taufwd - \taubkwd}
	\\ &=
	\left( 1 - \frac{2}{\taufwd - \taubkwd + 1} \right)^{\taufwd - \taubkwd}
	\\ &\approx
	\exp(-2).
\end{align*}
This motivates our use of $\operatorname{D}$ nearby $0.135$.

\begin{figure}[t]
\centering
\ifarxiv
\includegraphics[width=0.45\linewidth]{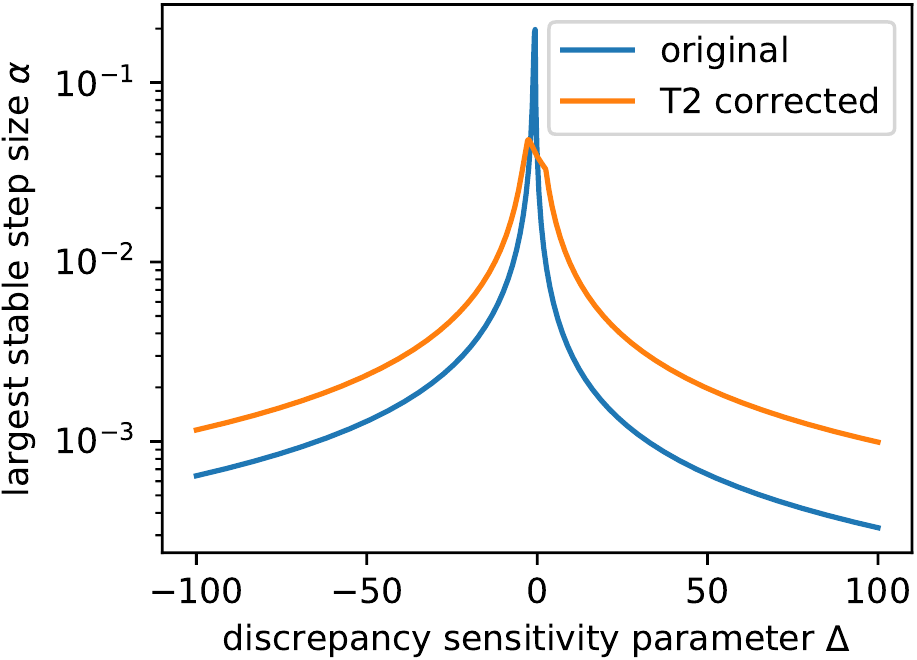}
\else
\includegraphics[width=0.66\linewidth]{figures/corr_comparison_t40_s30.pdf}
\fi
\caption{Plot of the largest step size $\alpha$ for which all the eigenvalues of the companion matrix lie within the unit disk for various values of the discrepancy sensitivity parameter $\Delta$, comparing the original quadratic model with the T2 discrepancy-corrected model. This figure was generated for $\taufwd = 40$ and $\taubkwd = 10$.} 
\label{fig:corrcomparison4030}
\end{figure}

In \Cref{sec:delay_discrepancy}, we also claimed that using T2 with the assignment in (\ref{eqnGammaValue}) seems to increase the allowable range over which the system is stable.
In experiments on the quadratic model, we observed that this happens consistently for all values of $\Delta > 0$ and for all $\taufwd$ and $\taubkwd$ we tried.
We tried all values of $\taufwd > \taubkwd$ where $\taufwd \le 40$ and values of $\Delta$ ranging from $-100$ to $100$; this range of $\tau$ covers the entire range of delays present in our DNN training experiments. 
While the improvement seems to happen always for $\Delta \ge 0$, if $\Delta < 0$ we have observed (again only in numerical experiments) that T2 does not necessarily improve the threshold of stability for all values of $\Delta$. 
This is illustrated in Figure~\ref{fig:corrcomparison4030}, which shows what happens for the particular case of $\taufwd = 40$ and $\taubkwd = 10$.
This figure is generally representative of what happens the cases we tried: the T2 correction makes the range of stable $\alpha$ consistently bigger when $\Delta \ge 0$, while occasionally having a negative effect when $\Delta \le 0$.
\section{Supplementary material for~\Cref{sec:experiments}}
\label{app:exp_app}
In this section, we discuss the setup details and additional experiment results. We first discuss the setup of each task we consider and the hyperparameter configuration of PipeMare in~\Cref{app:model_dataset}. We then present experiment results in addition to the performance and ablation study in \Cref{sec:experiments}.

\begin{table}
	\centering
	\small
	\begin{tabular}{c c c}
	\toprule
	Dataset & CIFAR10 & ImageNet \\
	\midrule
	Optimizer & \multicolumn{2}{c}{SGD with Momentum} \\
	Initial learning rate $\alpha$ &  $0.01$ & $0.1$\\
	Learning rate drop interval (epochs) & $80$ & $30$ \\
	Learning rate drop factor & $0.1$ & $0.1$ \\
	Momentum & 0.9 & 0.9 \\
	Training epochs & $200$ & $100$ \\
	$\l2$ regularization & $0.0005$ & $0.0001$ \\
	Minibatch size & $64$ & $256$ \\
	Microbatch size & $8$ & $16$ \\
	\bottomrule	
	\end{tabular}
	\caption{Training hyperparameters for ResNet 50 on CIFAR10 and ImageNet.}
	\label{tab:resnet}	
\end{table}

\subsection{Experiment setup}
\label{app:model_dataset}

We discuss the details in setup for each task we consider as well as in the hyperparameter configurations for PipeMare.

\paragraph{ResNet experiments.}
We use a publicly available implementation\footnote{https://github.com/kuangliu/pytorch-cifar} of ResNet for CIFAR10 which is reported to have good performance on CIFAR. We inherit the hyperparameters from the code repository except the initial learning rate. As the test accuracy associated with the provided learning rate does not reach $94.0$, we search it with grid $\{0.001, 0.01, 0.1\}$ to ensure the strong performance of synchronous baselines. We then uniformly apply the optimal value $0.01$ to all the synchronous and asynchronous pipeline-parallel training. For the ImageNet experiment, we fully inherit the model and training configurations from the official PyTorch implementation.\footnote{https://pytorch.org/} For both the CIFAR10 and ImageNet dataset, we use the standard train/validation/test dataset split in the Python Torchvision library. We present the detailed model hyperparameters and training configuration in Table~\ref{tab:resnet}.

\begin{table}
	\centering
	\small
	\begin{tabular}{c c c}
	\toprule
	Dataset & IWSLT & WMT \\
	\midrule
	Optimizer & \multicolumn{2}{c}{AdamW} \\
	Max learning rate & \num{5e-4} & \num{7e-4} \\
	Label smoothing & \multicolumn{2}{c}{$0.1$} \\
	Dropout & $0.3$ & $0.1$\\
	Weight decay & \num{1e-4} & $0$ \\
	LR linear warmup minibatches & \multicolumn{2}{c}{$8000$} \\
	Initial LR for linear warmup up & \multicolumn{2}{c}{\num{1e-7}} \\ 
	Adam $\beta$s & \multicolumn{2}{c}{$(0.9, 0.98)$} \\
	Training epochs & $60$ & $80$ \\
	Minibatch size (average \# of tokens) & $3600$ & $29000$ \\
	Microbatch size (max \# of tokens) & $245$ & $1792$ \\
	\# of microbatches & \multicolumn{2}{c}{$19$} \\
	Gradient norm clipping threshold & $25$ & NA \\
	\bottomrule	
	\end{tabular}
	\caption{Training hyperparameters for the Transformer on IWSLT and WMT. Here, ``LR'' stands for learning rate.}
	\label{tab:transformer}	
\end{table}

\begin{table*}
	\centering
	\small
	\begin{tabular}{c c c c}
	\toprule
	Dataset & Hyperparameters & Tuning grid & \makecell{Retuning grid for \\$\#$ of annealing epochs} \\
	\midrule
	\multirow{2}{*}{CIFAR10} & Number of annealing epochs (PipeMare T1) & $\{10, \mathbf{20}, 40, 80, 160\}$ & -- \\
							 & Discrepancy correction decay (PipeMare T1 + T2) & $\{0.1, \mathbf{0.5}, 0.9\}$ & $\{10, \mathbf{20}, 40\}$ \\
	\midrule
	\multirow{2}{*}{IWSLT14} & Number of annealing epochs (PipeMare T1) & $\{15, \mathbf{30}, 60\}$ & -- \\
							 & Discrepancy correction decay (PipeMare T1 + T2) & $\{0.01,\mathbf{0.1},0.2\}$ & $\{15, 20, \mathbf{30}\}$ \\
							 & Warmup epochs (PipeMare T1 + T2 + W) & $\{3, 5, \mathbf{10}\}$ & $\{1, \mathbf{10}, 20\}$ \\
	\bottomrule	
	\end{tabular}
	\caption{Hyperparameter sweep for PipeMare to demonstrate the best model accuracy attained by PipeMare. We sweep the number of annealing epochs, the discrepancy correction decay and the number of warmup epochs sequentially. For each hyperparameter, we first sweep it with optimal values for previously sweeped hyperparameters if there are any. After we tune the decay and number of warmup epochs, we also re-sweep the number of annealing epochs; we found this re-sweep can be important to model accuracy in cases such as PipeMare T1 + T2 + W for IWSLT. We use $0$ warmup epochs for CIFAR10 as we found warmup epochs does not improve the model accuracy. We bold the hyperparameter values attaining the best model accuracy in each grid.}
	\label{tab:pipemare_hyper}	
\end{table*}

\paragraph{Transformer experiments.}
We use the Fairseq implementation for 12-layer transformer models and inherit the key hyperparameters from the Fairseq repository.\footnote{Fairseq repo: https://github.com/pytorch/fairseq} We use $2\times$ longer learning rate linear warmup steps than in the original code repository across experiments because we observe $2\times$ linear warmup steps can produce higher BLEU scores for both the synchronous and asynchronous runs. For both the IWSLT14 and WMT17 German to English dataset, we use beam width $5$ to evaluate the BLEU score. We present the other hyperparameters in \Cref{tab:transformer} for reproducibility.

\paragraph{Hyperparameter of PipeMare.}
PipeMare has three key hyperparameters for the three techniques: the number of annealing epochs for learning rate rescheduling (T1); the decay $D$ for discrepancy correction (T2); the number of epochs (steps) for warmup epochs (W). 
To compare the best model accuracy attained by different training algorithms, we following the approach used by~\citet{wilson2017marginal}---we report the best test set model accuracy attained across the hyperparameter grid. 
For the CIFAR10 and IWSLT14 experiments, we sweep the annealing epochs, the decay and the number of epochs sequentially. When sweep each of these parameters, we first anchor on the optimal values of the already sweeped hyperparameters. We then re-sweep the number of annealing epochs after sweeping the grid for the decay and the number of warmup epochs; we observe this re-sweep on the number of annealing epochs can improve the model accuracy attained by PipeMare on IWSLT14. Note for each hyperparameter configuration, we report the model accuracy as the best performance across all training epochs.

In \Cref{tab:pipemare_hyper}, we present the hyperparameter grid we use as well as the optimal values (in bold) when sequentially sweeping the hyperparameters for CIFAR10 and IWSLT14 in \Cref{sec:experiments}. Note for CIFAR10, we found that warmup epochs do not further improve the statistical efficiency; we thus use 0 warmup epochs to attain the best performance on CIFAR10. To avoid the intensive computational overhead of tuning ImageNet and WMT, we transfer the three key hyperparameters of PipeMare from CIFAR10 and IWSLT with minimal search centered around them. Specifically, for ImageNet we use the same discrepancy correction as CIFAR10 and 10 epochs (one third of total epochs before base learning rate decayed by 10, note CIFAR10 uses 20 epochs, which is a quarter of the total epochs before learning rate decay) as annealing epochs. For WMT we used the same discrepancy correction as IWSLT and 4 epochs (16k minibatch steps, while IWSLT14 uses 12k minibatch steps) for synchronous warmup and another 4 epochs for annealing (IWSLT14 uses same epochs for synchronous warmup and annealing epochs as well). Following the optimal hyperparameter setting in IWSLT, we also use the same number of epochs for annealing epochs and warmup epochs for WMT; these PipeMare configurations for ImageNet and WMT are presented in \Cref{tab:pipemare_hyper_large}.

\subsection{Additional experiment results}
We present the additional experiment results in addition to the demonstration in \Cref{sec:experiments}. We discuss the results on ImageNet and WMT dataset in \Cref{app:mom_correction}. We then discuss supplementary results for PipeMare ablation study in \Cref{app:ablation}.

\subsubsection{ImageNet and WMT results}
\label{app:mom_correction}
In \Cref{sec:empirical_eval}, we discussed the end-to-end comparison on the ImageNet and WMT dataset. To better compare the statistical and hardware efficiency across pipeline training methods, we visualize the model accuracy as a function of number of epochs and of normalized time in \Cref{fig:imagenet}. For the ImageNet dataset, we can observe in \Cref{fig:imagenet} that PipeMare attains higher test accuracy than PipeDream. For the WMT dataset in \Cref{fig:imagenet}, PipeMare can attain competitive test BLEU score to GPipe synchronous results while PipeDream only demonstrate $0.0$ BLEU score.

\subsubsection{PipeMare ablation study}
\label{app:ablation}

\begin{figure*}[t]
\includegraphics[width=\linewidth]{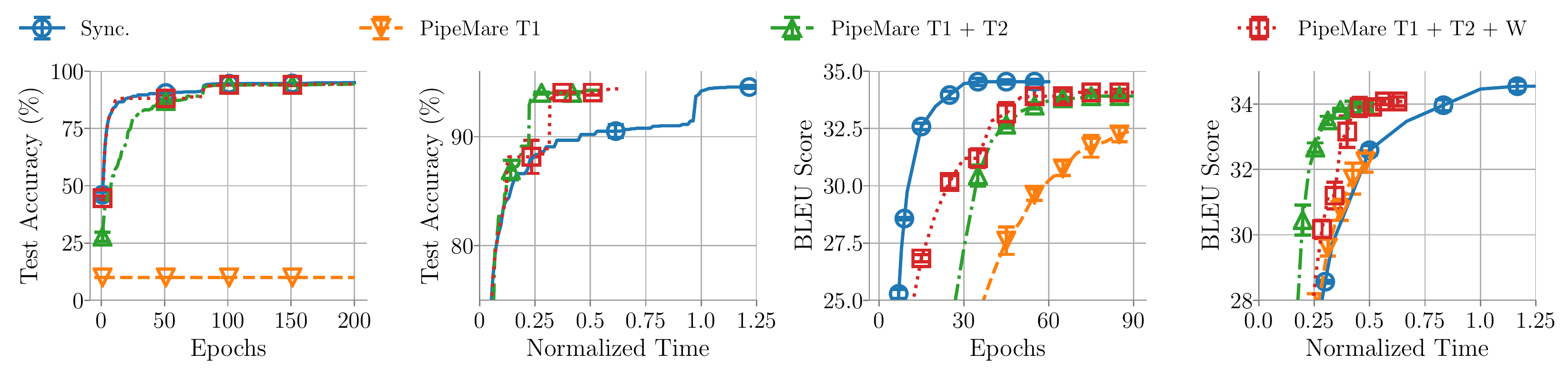}
\vspace{-8mm}

\caption{Impact of incrementally combining PipeMare techniques (T1, T2, and W) on a ResNet50 using Cifar10 (leftmost two figures) and 12-layer Transformer model (rightmost two figures) using IWSLT14 with $214$ and $186$ pipeline stages respectively.
This is 2x the number of pipeline stages when each model weight is treated at its own stage (as is done in \Cref{sec:experiments}). This tests the limits of our approach at an extreme (a fine-granularity of PP). Normalized time is computed using pipeline utilization and number of epochs, providing a proxy for idealized time on an accelerator.}
\vspace{-3mm}
\label{fig:trade-off-2x}
\end{figure*}

\paragraph{Ablation study: a different number of pipeline stages.} In \Cref{sec:ablation}, we perform ablation study with $107$ and $93$ pipeline stages respectively for CIFAR10 and IWSLT. In \Cref{fig:trade-off-2x} we demonstrate the ablation study with $214$ and $186$ stages for CIFAR10 and IWSLT. We observe that the learning rate rescheduling, discrepancy correction and warmup epochs can demonstrate even larger contributions to both the statistical and hardware efficiency as the number of pipeline stages is 2x larger here versus those shown in \Cref{sec:ablation}.

\begin{table}
	\centering
	\small
	\begin{tabular}{c c c}
	\toprule
	Dataset & ImageNet & WMT \\
	\midrule
	Sync warmup epochs & 30 & 4 \\
 	Discrepancy correction & $0.5$ & $0.1$ \\
	Annealing epochs & $10$ & $4$ \\
	\bottomrule	
	\end{tabular}
	\caption{PipeMare hyperparameters on the ImageNet and WMT dataset.}
	\label{tab:pipemare_hyper_large}	
\end{table}

\paragraph{Discrepancy correction for ResNet 152.}
In \Cref{sec:ablation}, we demonstrate discrepancy correction (T2) can improve the model accuracy on ResNet 50 and Transformer for CIFAR10 and IWSLT. In this section, we demonstrate that discrepancy correction can also contribute to preventing divergence for models with larger number of stages. More concretely, in \Cref{fig:resnet152}, we show that PipeMare T1 (only with learning rate rescheduling) diverge for ResNet 152 on CIFAR10 with $150$ pipeline stages. By additionally applying discrepancy correction, we observe that PipeMare converges and achieve matching test accuracy to GPipe training in a fixed number of epochs after the first learning rate drop after $80$ epochs.

\begin{figure*}
\includegraphics[width=\linewidth]{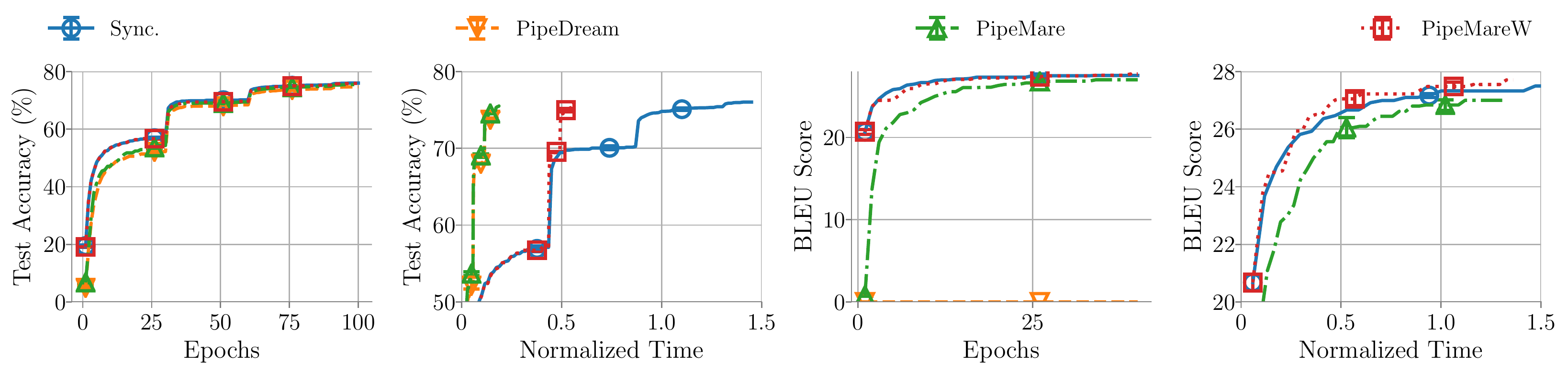}
\caption{The statistical performance and normalized time attained by different pipeline training methods on ImageNet (leftmost two figures) and WMT (rightmost two figures). We observe PipeMareW can attain higher model accuracy for both ImageNet and WMT, while being competitive to GPipe in the same number of epochs.  We also show that both PipeMare and PipeMareW achieve time-to-accuracy speedups (w.r.t. normalized time) over GPipe while PipeDream fails to converge and attains BLEU score close to 0 on Transformer model. On ImageNet PipeMare outperforms PipeDream in terms of time-to-accuracy while PipeMareW attains state-of-the-art accuracy on this task (which PipeDream and PipeMare do not).}	
\label{fig:imagenet}
\end{figure*}

\begin{figure*}
\centering
\ifarxiv
\includegraphics[width=0.4\linewidth]{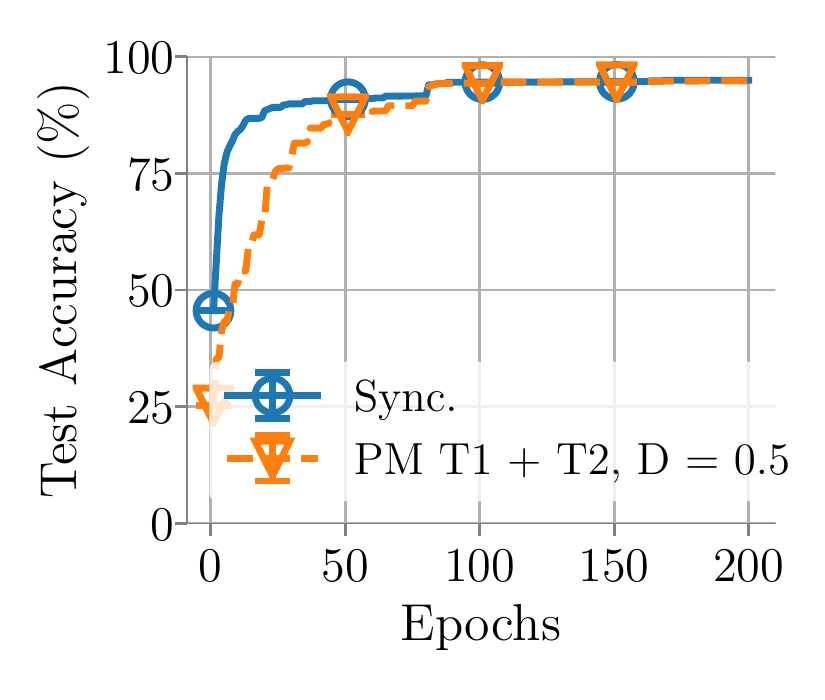}	
\else
\includegraphics[width=0.8\linewidth]{figures/resnet152.pdf}	
\fi
\caption{We observe ResNet 152 on CIFAR20 diverges when only using learning rate rescheduling (T1). Discrepancy correction is necessary to prevent divergence for ResNet 152 on CIFAR10; we observe PipeMare with discrepancy correction (T1 + T2) attains matching performance to GPipe synchronous training.}
\label{fig:resnet152}
\end{figure*}

\subsubsection{Hyperparameter sensitivity studies}
We empirically demonstrate the sensitivity of model accuracy to the three key hyperparameters in PipeMare. 

\paragraph{Sensitivity to annealing epochs.} One key hyperparameter for improving convergence using Heuristic 1 is the number of annealing epochs $K$. We further study here the sensitivity of model accuracy (loss) with respect to the number of annealing epochs in ResNet and Transformer model. As shown in Figure~\ref{fig:anneal_sens}, we observe that different model may require a different number of annealing epochs for optimal test performance. Specifically, we can see that the ResNet and Transformer model prefers small and large number of annealing epochs respectively.

\begin{figure*}
\centering
\begin{tabular}{c c}
	\includegraphics[width=0.49\linewidth]{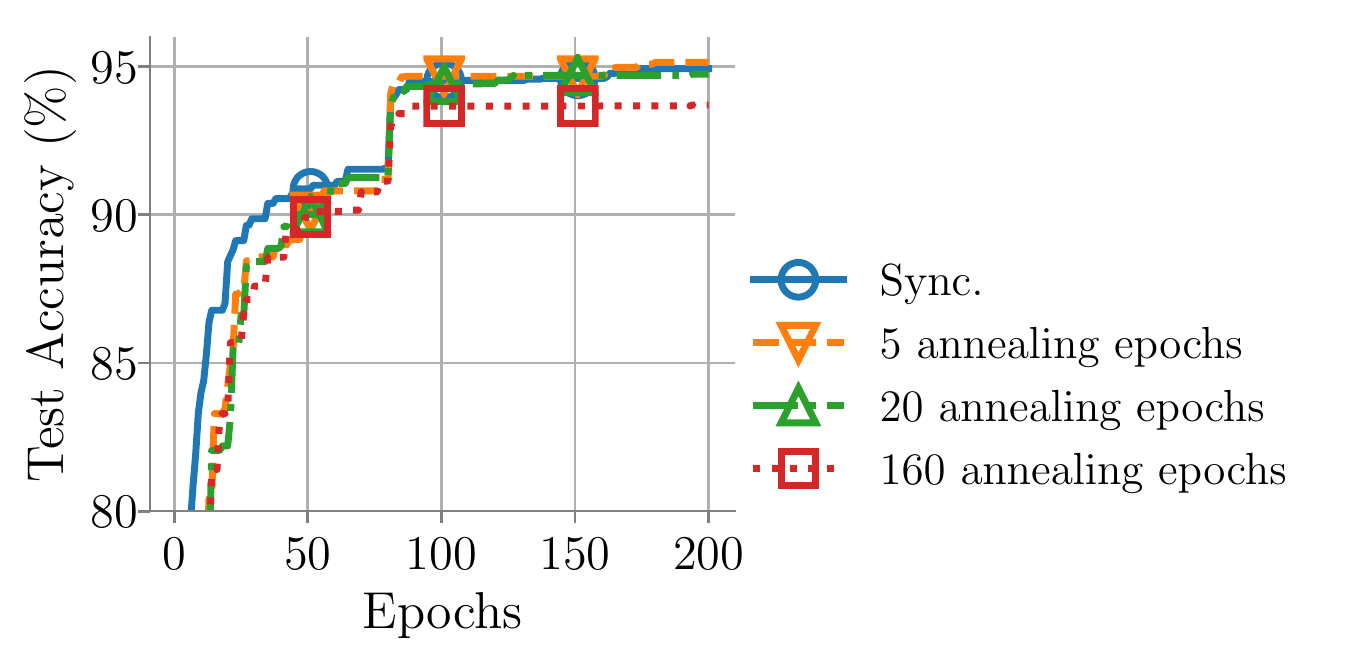} &
	\includegraphics[width=0.49\linewidth]{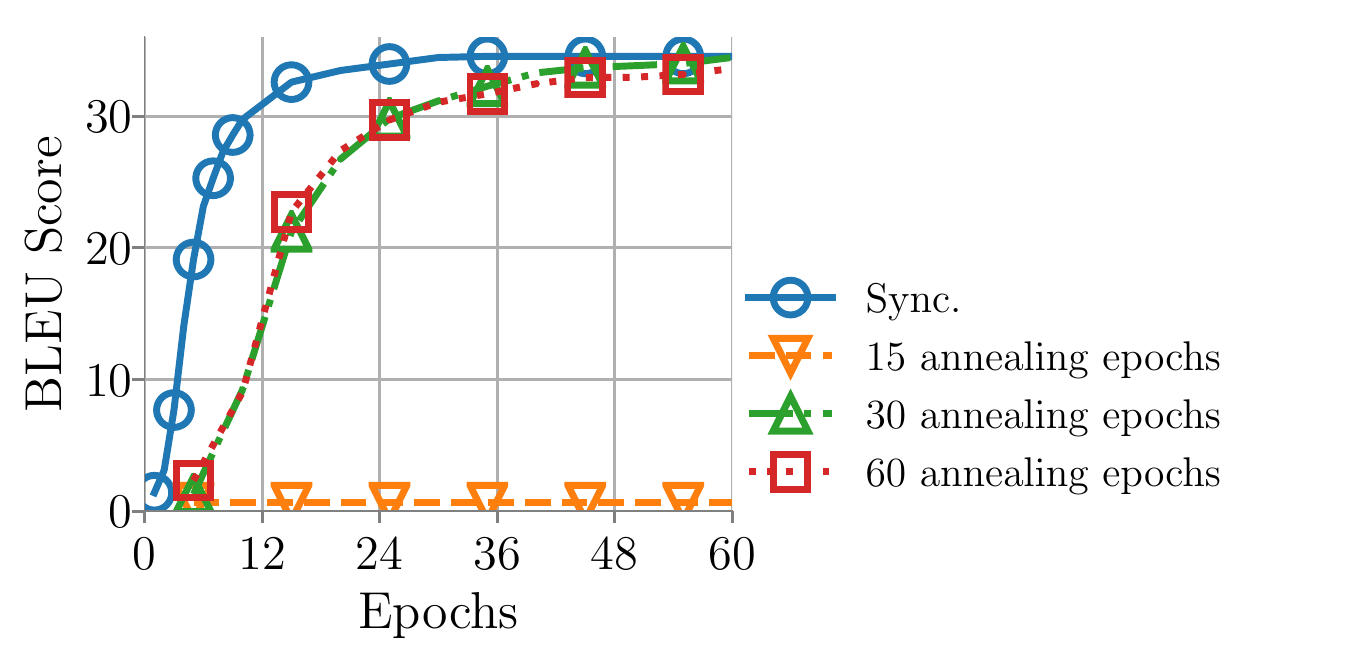}
\end{tabular}
\caption{Sensitivity of model accuracy to the number of annealing epochs. We observe that choosing the number of annealing epochs can be important to achieving model accuracy matching synchronous training.}
\label{fig:anneal_sens}	
\end{figure*}

\paragraph{Sensitivity to correction decay.} A right choice of correction decay is important to stabilizing the training and speed up the convergence. As shown in Figure~\ref{fig:decay_sens}, a proper correction decay $D$ ($\le 0.2$) can speed up the convergence of Transformer while an improper $D$ can result in even worse result than those without corrections. In other words, simply reusing the momentum buffer in SGD updates for correcting the parameters during backward could not fulfill the purpose of approximating the parameters used during forward. Therefore, an extra memory buffer and accumulation $\gamma$ is needed for each stage, which adds additional 25-33\% of memory to the total weight memory (e.g., in Adam, we have master weight, gradient, momentum, and norm, totally four copies of weight memory). 
\begin{figure*}
\centering
\begin{tabular}{c c}
	\includegraphics[width=0.49\linewidth]{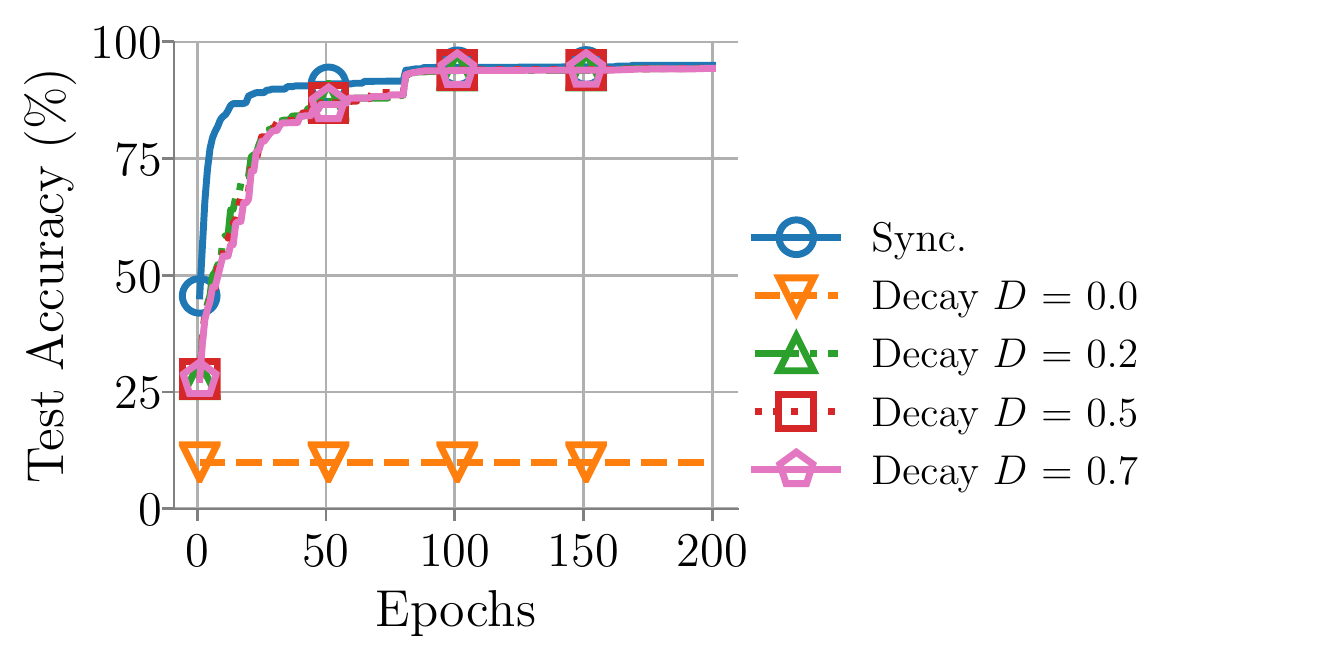} &
	\includegraphics[width=0.49\linewidth]{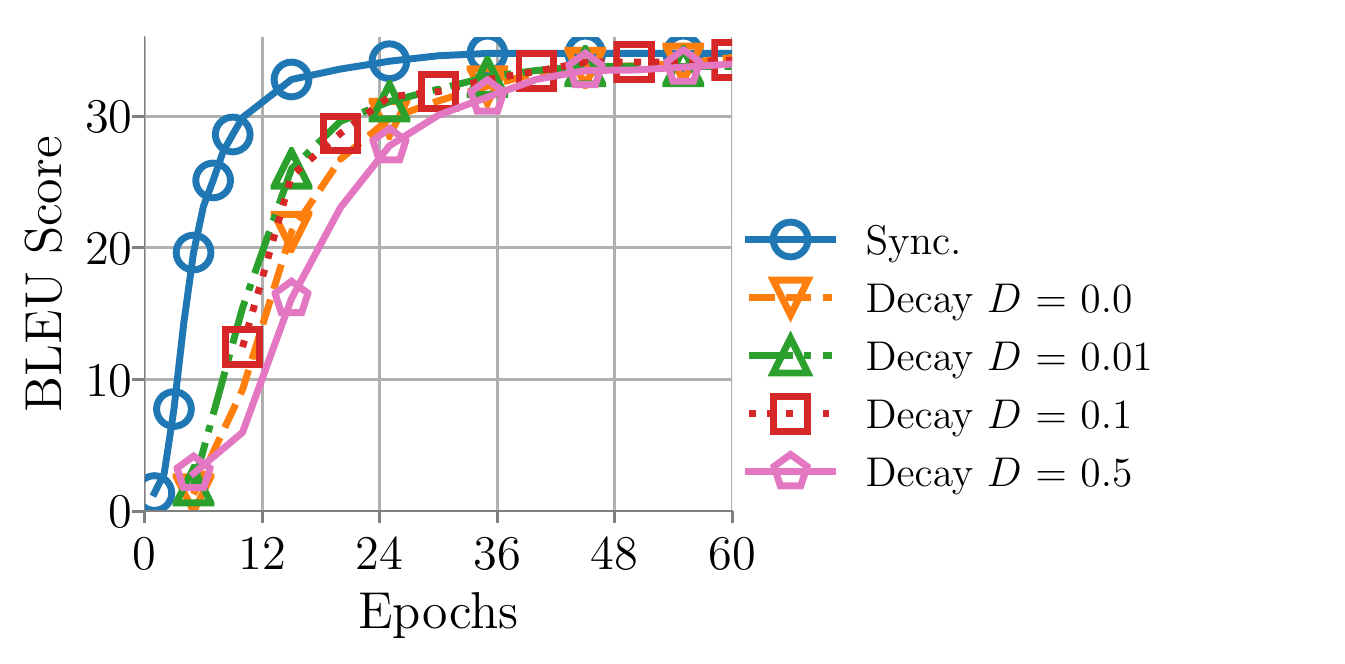}
\end{tabular}
\caption{Sensitivity of model accuracy to the decay $D$ for discrepancy correction. We notice that the decay value can have an impact on the convergence speed. For example, it requires a decay smaller than $0.5$ to converge faster than without discrepancy correction while $0.5$ can demonstrate test accuracy matching that attained by synchronous training.}
\label{fig:decay_sens}	
\end{figure*}

\begin{figure}
\ifarxiv
\centering
\includegraphics[width=0.65\linewidth]{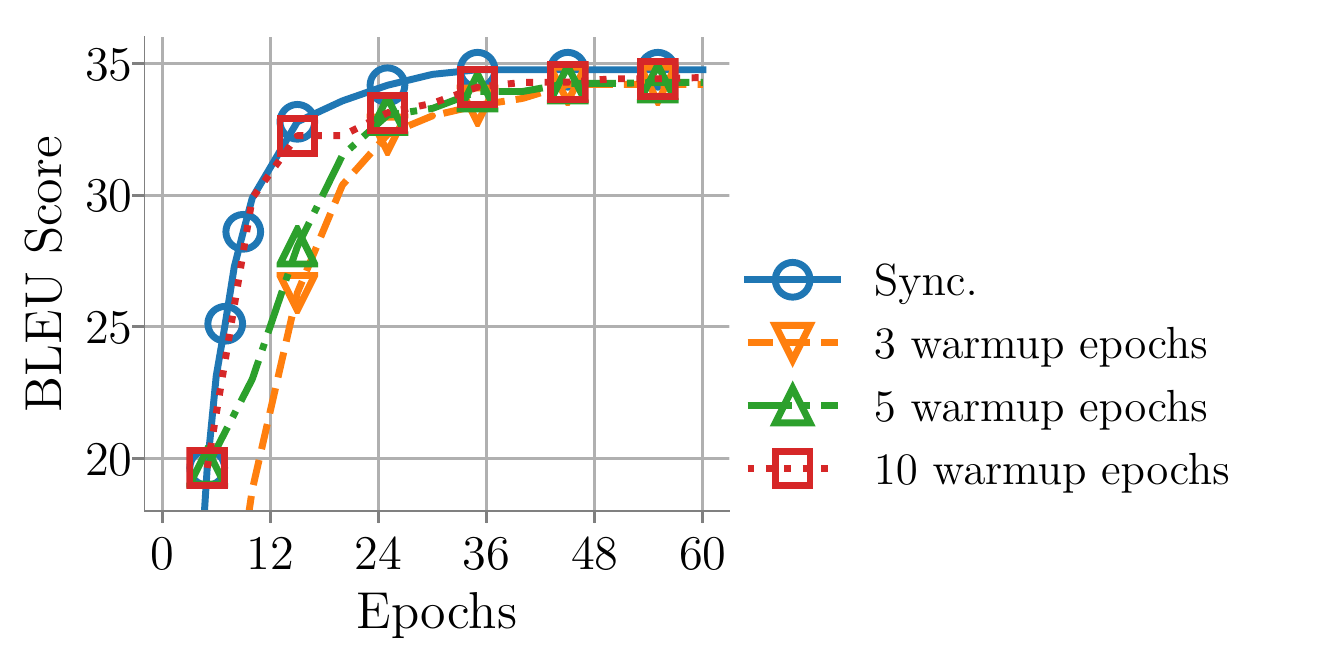}	
\else
\includegraphics[width=0.975\linewidth]{figures/exp3-fig1-iwslt-BLEU.pdf}	
\fi
\caption{Sensitivity of model accuracy to the number of synchronous warmup epochs on IWSLT. We observe a tradeoff in using warmup epochs: a large number of warmup epochs can harm the throughput but converges to a good model accuracy in fewer epochs.}
\label{fig:warmup_sensI}
\end{figure}

\paragraph{Sensitivity to warmup epochs.}  
In \Cref{fig:warmup_sensI}, we show the impact of different numbers of synchronous warmup training epochs on PipeMare's convergence. This exposes a tradeoff between statistical efficiency and hardware efficiency. More synchronous warmup epochs lower the overall pipeline utilization (and therefore throughput) but often help statistical convergence. 

\section{Statistical Efficiency and Recompute}
\label{app:recompute}
In pipeline-parallel training, to compute the gradient in a pipelined fashion, the activation memory needs to be stored for each batch of data at every pipeline stage. For fine-grained pipeline-parallel training, this can results in significantly increased memory footprint. To reduce the memory incurred by activations, the activation recomputation technique~\cite{chen2016training,huang2018gpipe} has been proposed for training deep neural networks. We first discuss the recomputation model in asynchronous pipeline-parallel training in~\cref{subsec:recompute_model}. We then demonstrate in~\cref{subsec:evaluation_res} that PipeMare with recomputation can attain matching / competitive model accuracy while using less memory footprint comparing to PipeMare without recomputation.

\subsection{Asynchronous pipeline-parallel recomputation}
\label{subsec:recompute_model}

When running with asynchronous pipeline parallelism, adding recompute adds additional delay paths to the computation, since now the backward pass depends not only on a single delayed weight value but also on delayed recomputed activations, each of which may have a different delay from the delay used for the backward-pass weights.
We can model this formally as
\[
	w_{t+1} = w_t - \alpha \nabla f_t(u_{\text{fwd},t}, u_{\text{bkwd},t}, u_{\text{recomp},t}),
\]
where now $u_{\text{recomp},t}$ denotes the delayed version of the weights used for recomputing activations in the backward pass for the $t$th gradient microbatch.
Just as for the other delayed weights, we define this in terms of a fixed delay as
\[
	\left( u_{\text{recomp},t} \right)_i = \left( w_{t - \tau_{\text{recomp},i}} \right)_i
\]
where now $\tau_{\text{recomp},i}$ is a fixed delay that affects weights used for recomputation in the $i$th layer.
Given this definition, there is a natural way we can extend the discrepancy correction of T2 to apply to these new recomputed activations.

\textbf{T2 for Recompute:}
\emph{Instead of the assignment of $u_{\text{recomp}}$ above, set
\[
	\left( u_{\text{recomp},t} \right)_i
	=
	\left( w_{t - \tau_{\text{recomp},i}} \right)_i
	-
	\left( \tau_{\text{fwd},i} - \tau_{\text{recomp},i} \right) \delta_{t, i},
\]
where $\delta_{t, i}$ is the same weight-trajectory accumulator used to correct $u_{\text{bkwd},t}$ in T2.}

\paragraph{The theory.}
To model delay discrepancy with recomputation in the quadratic model, we now assume gradient samples of the form
\begin{align*}
	\nabla f_t(u_{\text{fwd},t}, u_{\text{bkwd},t}) &= (\lambda + \Delta) \cdot w_{t-\tau_{\text{fwd}}} \\&\hspace{2em}- (\Delta - \Phi) \cdot w_{t-\tau_{\text{bkwd}}} \\&\hspace{2em}- \Phi w_{t-\tau_{\text{recomp}}} - \eta_t
\end{align*}
where $\tau_{\text{fwd}} > \tau_{\text{recomp}} > \tau_{\text{bkwd}}$ are now three different delays, and $\Phi$ is new a constant that measures the sensitivity of the gradients to discrepancy between the recomputed weights and the backward-pass weights.
As before, we can think of this as the natural first-order (linear) approximation of $\nabla f_t$; it can model any affine function of $u_{\text{fwd},t}$, $u_{\text{bkwd},t}$, and $u_{\text{recomp},t}$ that is consistent with the curvature $\lambda$ when $u_{\text{fwd},t} = u_{\text{bkwd},t}$.
If $\Phi = 0$, we recover our original no-recomputation setting, whereas for large-magnitude values of $\Phi$, even a small delay discrepancy in recomputation could cause a large effect on the gradient samples.

It is straightforward to see that the characteristic polynomial of the companion matrix here will be
\begin{align*}
	p(\omega)
	&=
	(\omega - 1) (\omega - \gamma) \omega^{\taufwd}
	\\&\hspace{-1em}+
	\alpha (\lambda + \Delta) (\omega - \gamma)
	\\&\hspace{-1em}-
	\alpha (\Delta - \Phi) \omega^{\taufwd - \taubkwd} (\omega - \gamma)
	\\&\hspace{-1em}+
	\alpha (\Delta - \Phi) \omega^{\taufwd - \taubkwd} (\taufwd - \taubkwd) (1 - \gamma) (\omega - 1)
	\\&\hspace{-1em}-
	\alpha \Phi \omega^{\taufwd - \taurecomp} (\omega - \gamma)
	\\&\hspace{-1em}+
	\alpha \Phi \omega^{\taufwd - \taurecomp} (\taufwd - \taurecomp) (1 - \gamma) (\omega - 1).
\end{align*}

\begin{figure}
	\centering
	\ifarxiv
	\includegraphics[width=0.7\linewidth]{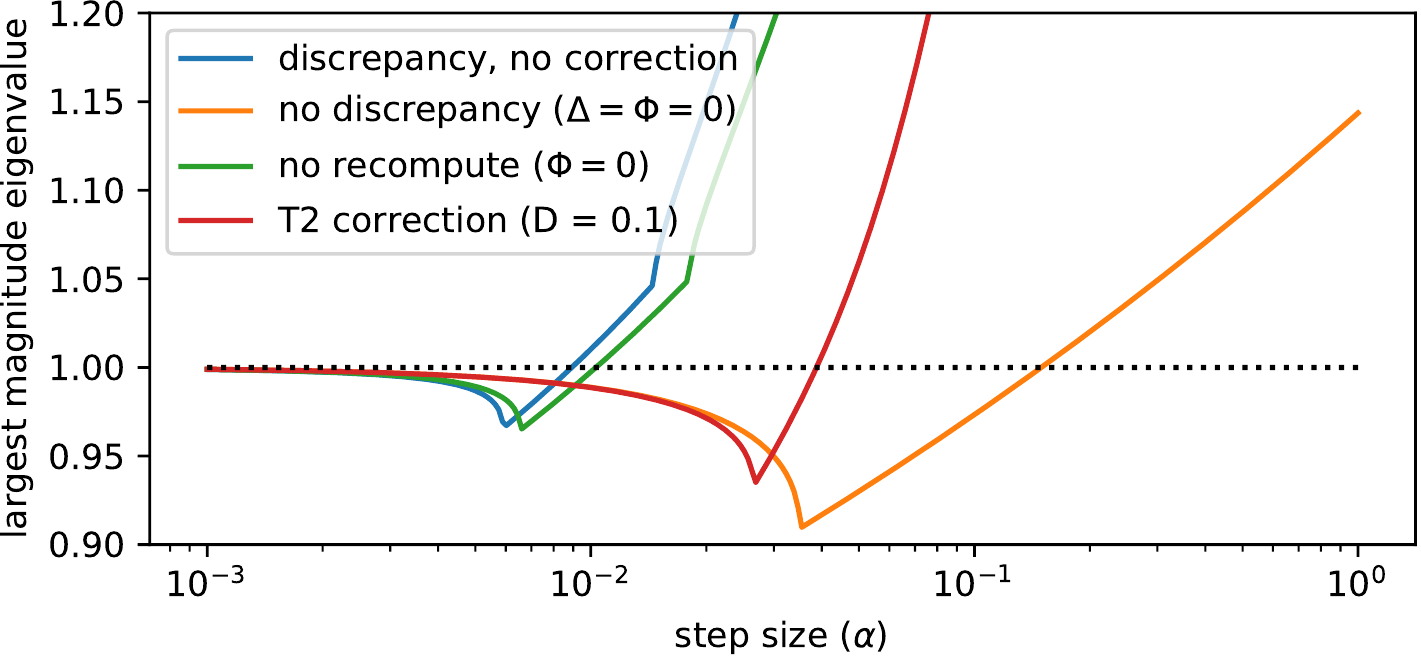}
	\else
	\includegraphics[width=0.9\linewidth]{figures/recomputecorr.pdf}
	\fi
	\caption{Effect of discrepancy correction on the quadratic model when recompute is used for a model with $\Delta = 10$, $\Phi = -5$, $\taufwd = 10$, $\taubkwd = 1$, $\taurecomp = 4$, and $\lambda = 1$. 
	Forward-backward delay discrepancy (blue) increases the largest magnitude eigenvalue of the companion matrix, just as in the no-recompute case (green).
	Discrepancy correction with $\operatorname{D} = 0.1$ (red) reduces the largest magnitude eigenvalue; this eigenvalue is closer to that attained without delay discrepancy (orange).}
	\label{fig:recomputecorrquadratic}
\end{figure}

While the complexity of this polynomial makes it difficult to prove a tight result like Lemma~\ref{lemmaRootsGD}, we can still analyze its spectral radius empirically, as we did for the non-recompute case in the main body of the paper.
Figure~\ref{fig:recomputecorrquadratic} shows this analysis.
Here we see that, just as in the case without recompute, delay discrepancy correction increases the range of step sizes over which the quadratic model is stable, and brings the behavior of the model closer to the no-delay-discrepancy case.

\begin{figure}
	\centering
	\ifarxiv
	\includegraphics[width=0.7\linewidth]{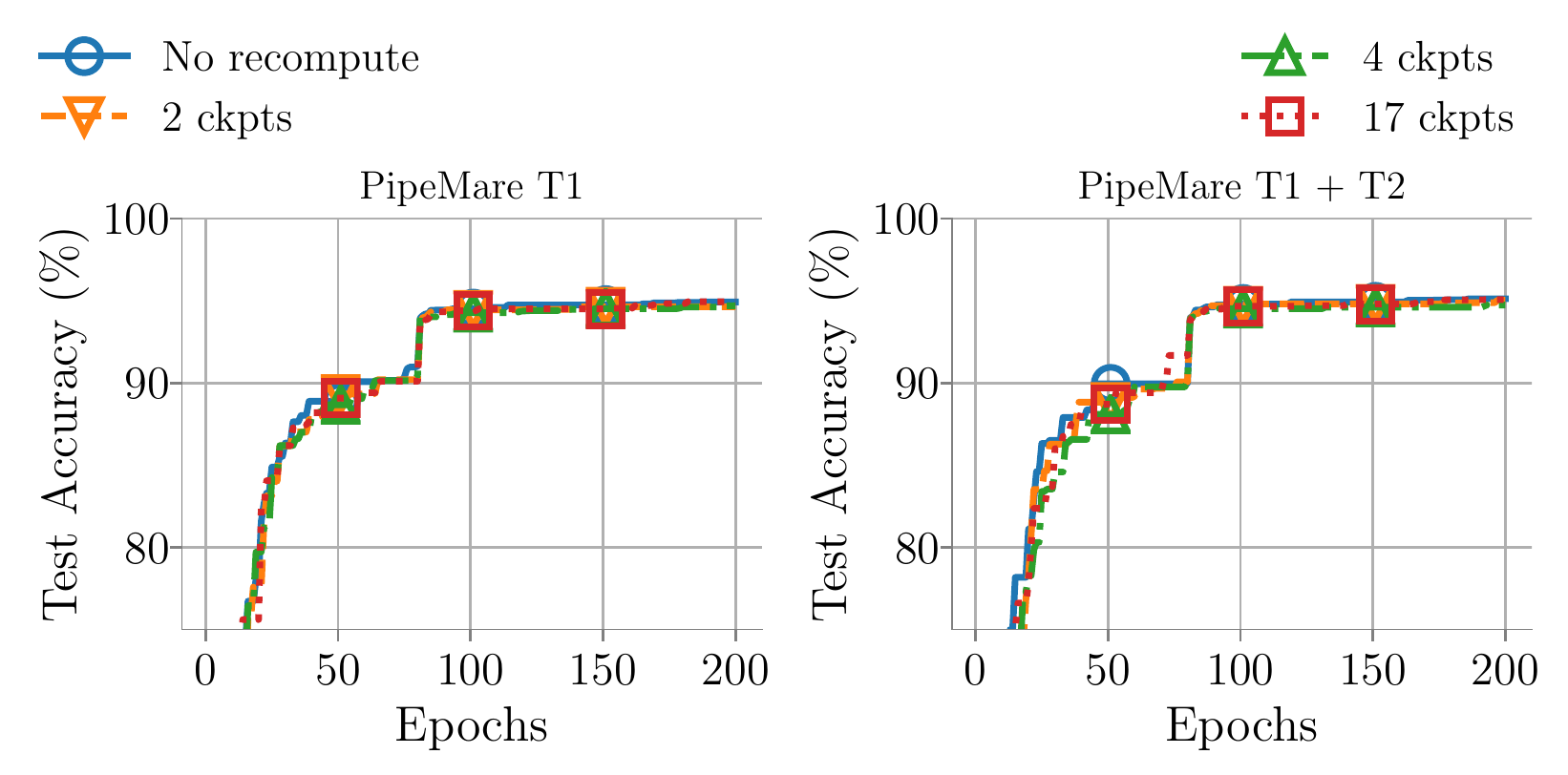}
	\else
	\includegraphics[width=0.975\linewidth]{figures/cifar_recompute_joint.pdf}
	\fi
	\caption{The statistical performance of recompute with different number gradient checkpoints on CIFAR10. We observe that with different number of gradient checkpoints, PipeMare with recompute can match the model accuracy attained by PipeMare without recompute. This indicates that recompute can significantly save the memory for storing activations with minimal influence on the attained model accuracy.}
	\label{fig:cifar_recompute}
\end{figure}

\begin{figure*}
	\centering
	\ifarxiv
	\includegraphics[width=0.9\linewidth]{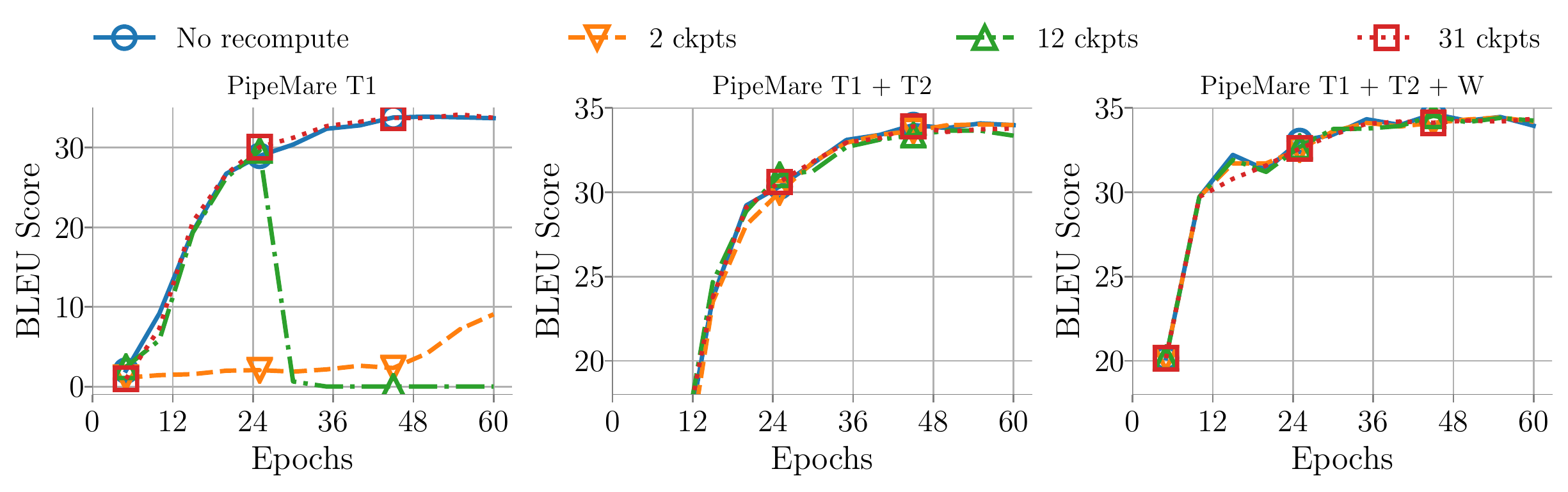}
	\else
	\includegraphics[width=0.75\linewidth]{figures/iwslt_recompute_joint.pdf}
	\fi
	\caption{The statistical performance of recompute with different number gradient checkpoints on IWSLT. We observed in the left plot that when only using learning rate rescheduling (T1) without discrepancy correction, recompute can unstable training with 2 and 12 gradient checkpoints. After applying discrepancy correction (T2) in the middle and right plots, we observe that with different number of gradient checkpoints, PipeMare with recompute can match the model accuracy attained by PipeMare without recompute. This indicates the importance of discrepancy correction to attaining stable recompute in PipeMare.}
	\label{fig:iwslt_recompute}
\end{figure*}

\subsection{Statistical efficiency and recompute}
\label{subsec:evaluation_res}
To study the impact of recompute over statistical efficiency, we study the model accuracy attained by PipeMare with recompute on CIAFR10 and IWSLT. We observe that 1) as discussed in \Cref{subsec:recompute_model}, discrepancy correction can be important to the stability of asynchronous training with recompute;  2) with different number of gradient checkpoints for recompute, PipeMare in general attains competitive or matching model accuracy to that attained by PipeMare without recompute.

\paragraph{Setup.}
In our experiment, we set gradient checkpoints at the natural module boundaries defined by skip connections. More concretely, ResNet uses residual connection between groups of convolutional layers while Transformer uses skip connections for both the multiple headed attention and feedforward modules. Following this principle, we use $\{2, 4, 17\}$ checkpoints and $\{2, 12, 31\}$ checkpoints to segment the models respectively for ResNet 50 and 12-layer Transformer model. To fully study the impact of recompute, we consider different combination of the key techniques in PipeMare. Specifically we consider T1, T1 + T2 and T1 + T2 + W for IWSLT; we consider only T1 and T1 + T2 for CIFAR 10 as warm up epochs (W) does not bring observable model accuracy improvement on CIFAR10.

\paragraph{Importance of discrepancy correction.}
In~\Cref{fig:cifar_recompute} and \Cref{fig:iwslt_recompute}, we plot the model accuracy attained by PipeMare with recompute using different number of gradient checkpoints on CIFAR10 and IWSLT. For the CIFAR10 case in \Cref{fig:cifar_recompute}, using recompute does not affect the model accuracy attained with discrepancy correction (PipeMare T1 + T2) and without discrepancy correction (PipeMare T1). However for the IWSLT case in \Cref{fig:iwslt_recompute}, without discrepancy correction (PipeMare T1), training with recompute in the asynchronous setting can be unstable. E.g. training with 2 gradient checkpoints fails to attain BLEU higher than $10.0$ while it diverge in the middle of training for 12 gradient checkpoints. When we apply the discrepancy correction in the middle and right plot of \Cref{fig:iwslt_recompute}, we can observe that PipeMare with different number of gradient checkpoints can achieve matching model accuracy to training without recompute. These observations indicate that discrepancy correction is important to the stability of training with recompute.

\paragraph{Statistical efficiency with recomputation.}
In~\Cref{fig:cifar_recompute} (right) and \Cref{fig:iwslt_recompute} (middle, right), we can see that with discrepancy correction, PipeMare asynchronous pipeline-parallel training can consistently attain strong model accuracy on both CIFAR10 and IWSLT. This further emphasizes that PipeMare can be orthogonally combined with recompute to attain strong model accuracy with significantly reduced activation memory footprint.

\section{Hogwild! asynchrony}
\label{app:hogwild_extension}

Asynchrony has been studied in various settings to accelerate the training of machine learning models~\cite{recht2011hogwild,kurth2017deep}. We ask the question of whether our proposed heuristic can go beyond the asynchronous pipeline setting with fixed gradient delay pattern, and accelerate training in classical asynchronous settings with stochastic gradient delay. 
In this section, we show that our learning rate rescheduling heuristic can also improve the model accuracy attained by training under the Hogwild!-style stochastic asynchrony~\cite{recht2011hogwild,de2015taming}. We first discuss the Hogwild!-style stochastic asynchrony model and then dive into the detailed experiment results.

\paragraph{Stochastic asynchrony model}
Hogwild!-style asynchrony considers a setting where the model is updated with a staled gradient. Specifically, the update of SGD algorithm over an objective function $f(w)$ can be written as 
\begin{equation}
	w_{t+1} = w_{t} - \alpha \nabla f_{t-\tau}(w_{t-\tau})
	\label{equ:orig_hogwild}
\end{equation}
where $w_t \in \mathbb{R}^d$ is the model iterate while $\nabla f_t(W_t)$ is the stochastic estimate of the gradient $\nabla f(w_t)$ at time step t. The $\tau_t$ here is a random variable describing the delay of the gradient; this random variable can model the delay of gradients due to the network transmission in distributed asynchronous training~\cite{kurth2017deep} or asynchronous model update in the shared memory settings~\cite{recht2011hogwild}.

We consider a variant of the original Hogwild!-style asynchrony model with different delays for different stages; this stage specific gradient delay setting is studied in our fixed delay asynchronous pipeline training in Section~\ref{sec:theory}. In particular, the model update for each stage can be characterized by 
\begin{equation}
	w_{i, t+1} = w_{i, t} - \alpha [\nabla f_{t-\tau_i}(w_{t-\tau_i})]_i
	\label{equ:variant_hogwild}	
\end{equation}
where $\tau_i$ is the stochastic gradient delay for the i-th stage and $[\nabla f_{t}(w_{t-\tau_i})]_i$ describes the gradient dimensions corresponding to the i-th stage.

In our variant of the Hogwild!-style gradient delay $\tau_i$, we sample from truncated exponential distributions following the existing study in asynchronous training~\cite{mitliagkas2016asynchrony}; this truncated exponential distribution is the maximum entropy distribution. We use the exponential distribution truncated at $\tau_{\max}$ uniformly for different stages to make sure we have bounded delay of the gradient. To model the different level of gradient delay for different stages, we use sampling distributions with different expectation values.

\begin{figure}
    \centering
    \ifarxiv
    \includegraphics[width=0.7\linewidth]{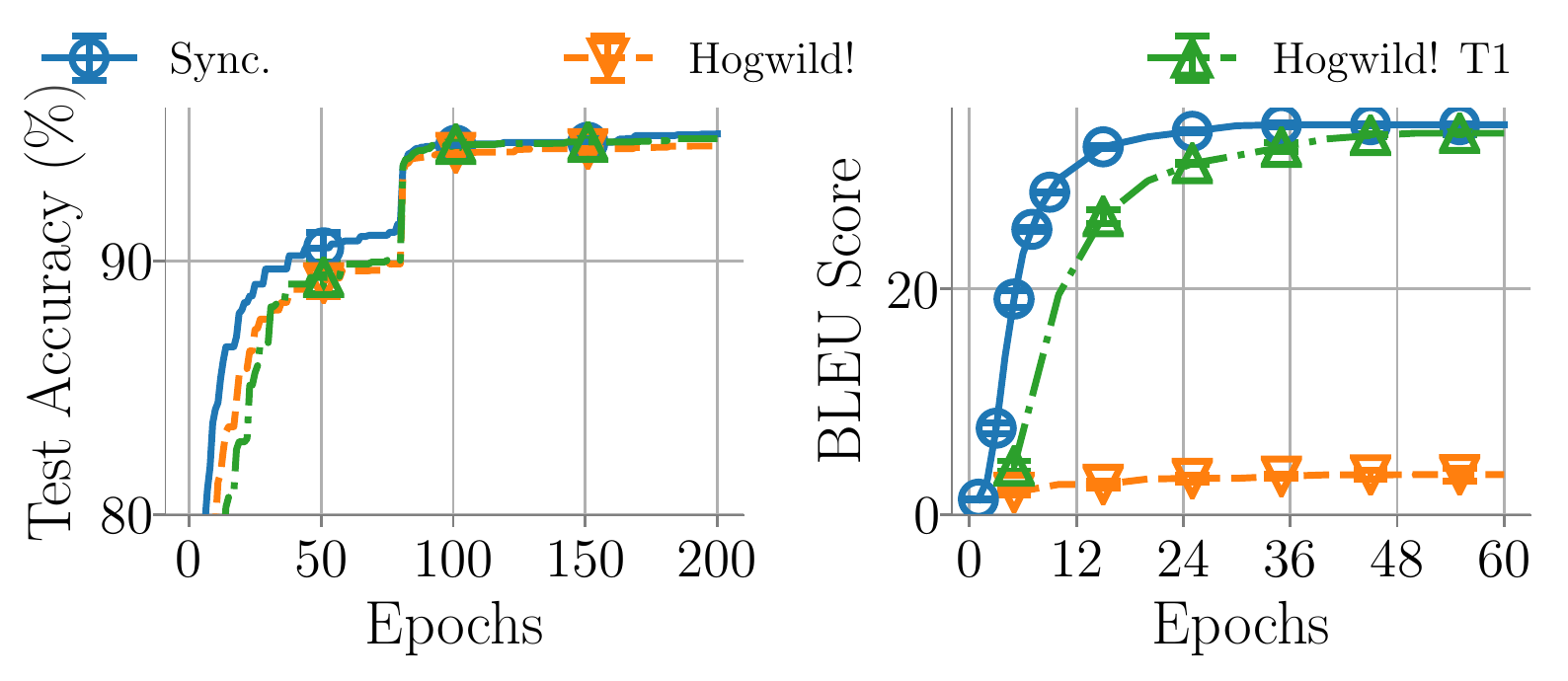}
    \else
    \includegraphics[width=\linewidth]{figures/hogwild-all-stat-efficiency.pdf}
    \fi
    \caption{Test performance of CIFAR10 ResNet (left) and IWSLT14 Transformer (right) under the Hogwild!-style asynchronous training. By using the learning rate rescheduling heuristic for asynchronous training, we can achieve test performance matching those attained by synchronous training. Comparing to asynchronous training without learning rate rescheduling, applying the rescheduling heuristic can attain better test performance after the same number of training epochs.}
    \label{fig:hogwild}
\end{figure}

\paragraph{Evaluation results}
To demonstrate that our learning rate rescheduling rule can also improve the model accuracy for training under Hogwild!-style asynchrony, we evaluate with the ResNet50 model on the CIFAR10 dataset and the Transformer model on the IWSLT14 German to English translation task. In our experiment, we use the maximal number of stages with at least one model weight in each group, which is also used in our pipeline training experiments in Section~\ref{sec:ablation}. Specifically, we use 107 and 93 stages for the ResNet and Transformer model respectively. We thus also inherit the optimal configuration for annealing epochs from the experiment on PipeMare only with learning rate rescheduling (PipeMare T1) in \Cref{sec:ablation}.
As shown in Figure~\ref{fig:hogwild}, we can observe that asynchronous training without learning rate rescheduling attains $94.51\%$ test accuracy and test BLEU score $3.6$ respectively for ResNet and Transformer. 
By applying learning rate rescheduling as described in Section~\ref{sec:lr_rescheduling}, we improve the test accuracy to $94.80\%$ and test BLEU score $33.8$ for asynchronous pipeline-parallel training for the ResNet and Transformer model. These observations indicates that our learning rate rescheduling heuristics can also improve the test performance of training under Hogwild!-style asynchrony.

\end{document}